\documentclass[11pt]{article}
\usepackage[utf8]{inputenc}
\usepackage{style}
\usepackage{framed}
\usepackage[parfill]{parskip}
\usepackage{setspace}
\usepackage{hyperref}

\allowdisplaybreaks

\title{Mildly Exponential Time Approximation Algorithms for Vertex Cover, Uniform Sparsest Cut and Related Problems}

\author{
Pasin Manurangsi\thanks{Email: \texttt{pasin@berkeley.edu}.} \vspace{-0.5em}\\
UC Berkeley
\and
Luca Trevisan\thanks{Email: \texttt{luca@berkeley.edu}.} \vspace{-0.5em}\\
UC Berkeley
}

\begin{document}

\maketitle

\begin{abstract}
In this work, we study the trade-off between the running time of approximation algorithms and their approximation guarantees. By leveraging a structure of the ``hard'' instances of the Arora-Rao-Vazirani lemma~\cite{ARV09,Lee05}, we show that the Sum-of-Squares hierarchy can be adapted to provide ``fast'', but still exponential time, approximation algorithms for several problems in the regime where they are believed to be NP-hard. Specifically, our framework yields the following algorithms; here $n$ denote the number of vertices of the graph and $r$ can be any positive real number greater than 1 (possibly depending on $n$). 
\begin{itemize}
\item A $\left(2 - \frac{1}{O(r)}\right)$-approximation algorithm for Vertex Cover that runs in $\exp\left(\frac{n}{2^{r^2}}\right)n^{O(1)}$ time. 
\item An $O(r)$-approximation algorithms for Uniform Sparsest Cut, Balanced Separator, Minimum UnCut and Minimum 2CNF Deletion that runs in $\exp\left(\frac{n}{2^{r^2}}\right)n^{O(1)}$ time.
\end{itemize}
Our algorithm for Vertex Cover improves upon Bansal \etal's algorithm~\cite{BCLNN17} which achieves $\left(2 - \frac{1}{O(r)}\right)$-approximation in time $\exp\left(\frac{n}{r^r}\right)n^{O(1)}$. For the remaining problems, our algorithms improve upon $O(r)$-approximation $\exp\left(\frac{n}{2^r}\right)n^{O(1)}$-time algorithms that follow from a work of Charikar \etal~\cite{CMM10}.
\end{abstract}

\section{Introduction}

Approximation algorithms and fast (sub)exponential time exact algorithms are among the two most popular approaches employed to tackle NP-hard problems. While both have had their fair share of successes, they seem to hit roadblocks for a number of reasons; the PCP theorem~\cite{AS98,ALMSS98} and the theory of hardness of approximation developed from it have established, for many optimization problems, that trivial algorithms are the best one could hope for (in polynomial time). On the other hand, the Exponential Time Hypothesis (ETH)~\cite{IP01,IPZ01} and the fine-grained reductions surrounding it have demonstrated that ``brute force'' algorithms are, or at least close to, the fastest possible for numerous natural problems.

These barriers have led to studies in the cross-fertilization between the two fields, in which one attempts to apply both techniques simultaneously to overcome known lower bounds. Generally speaking, these works study the trade-offs between the running time of the algorithms and the approximation ratio. In other words, a typical question arising here is: what is the best running time for an algorithm with a given approximation ratio $\tau$?

Optimization problems often admit natural ``limited brute force'' approximation algorithms that use brute force to find the optimal solution restricted to a subset of variables and then extend this to a whole solution. Similar to the study of fast exact algorithms for which a general motivating question is whether one can gain a noticeable speedup over ``brute force'',
the analogous question when dealing with approximation algorithms is whether one can do significantly better than these limited brute force algorithms. 

For example, let us consider the E3SAT problem, which is to determine whether a given 3CNF formula is satisfiable. The brute force (exact) algorithm runs in $2^{O(n)}$ time, while ETH asserts that it requires $2^{\Omega(n)}$ time to solve the problem. The optimization version of E3SAT is the Max E3SAT problem, where the goal is to find an assignment that satisfies as many clauses as possible. On the purely approximation front, a trivial algorithm that assigns every variable uniformly independently at random gives 7/8-approximation for Max E3SAT, while Hastad's seminal work~\cite{Has01} established NP-hardness for obtaining $(7/8 + \varepsilon)$-approximation for any constant $\varepsilon > 0$. The ``limited brute force'' algorithm for Max E3SAT chooses a subset of $O(\epsilon n)$ variables, enumerates all possible assignments to those variables and picks values of the remaining variables randomly; this achieves $(7/8 + \varepsilon)$-approximation in time $2^{O(\epsilon n)}$. Interestingly, it is known that running time of $2^{\Omega(\poly(\varepsilon)n)}$ is necessary to gain a $(7/8 +\varepsilon)$-approximation if one uses Sum-of-Squares relaxations~\cite{Gri01,Sch08,KMOW17}, which gives some evidence that the running time of ``limited brute force'' $(7/8 + \varepsilon)$ approximation algorithms for Max E3SAT are close to best possible.

In contrast to Max E3SAT, one can do much better than ``limited brute force'' for Unique Games. Specifically, Arora et al.~\cite{AIMS10} show that one can satisfy an $\epsilon$ fraction of clauses in a $(1-\epsilon)$-satisfiable instance of Unique
Games in time $2^{n/exp(1/\epsilon)}$, a significant improvement over the trivial $2^{O(\varepsilon n)}$ time ``limited brute force'' algorithm. This algorithm was later improved by the celebrated algorithm of Arora, Barak and Steurer~\cite{ABS15}  that runs in time $2^{n^{poly(\epsilon)}}$.

A number of approximation problems, such as $(2 - \varepsilon)$-approximation of Vertex Cover~\cite{KR08,BK09}, $(0.878\cdots + \varepsilon)$ approximation of Max Cut~\cite{KKMO07}, and constant approximation of Non-uniform Sparsest Cut~\cite{CKKRS06,KV15} are known to be at least as hard as Unique Games, but are not known to be equivalent to Unique Games. If they were equivalent, the subexponential algorithm of \cite{ABS15} would also 
extend to these other problems. It is then natural to ask whether these problems admit subexponential time algorithms, or at least ``better than brute force'' algorithms. Indeed, attempts have been made to design such algorithms~\cite{ABS15,GS11}, although these algorithms only achieve significant speed-up for specific classes of instances, not all worst case instances.

Recently, Bansal et al. \cite{BCLNN17} presented a ``better than brute force'' algorithm for Vertex Cover, which achieve a $(2 - 1/O(r))$-approximation in time $2^{O(n / r^r)}$. Note that the trade-off between approximation and running time is more analogous to the \cite{AIMS10} algorithm for Unique Games than with the ``limited brute force'' algorithm for Max 3ESAT discussed above.

The algorithm of Bansal et al. is partially combinatorial and is based on a reduction to the Vertex Cover problem in bounded-degree graphs, for which better approximation algorithms are known compared to general graphs. Curiously, the work of Bansal et al. does not subsume the best known polynomial time algorithm for Vertex Cover: Karakostas \cite{Kar09} shows that there is a polynomial time algorithm for Vertex Cover that achieves a $\left(2 - \frac {\Omega(1)}  {\sqrt{\log n}}\right)$ approximation ratio, but if one set $r:= \sqrt {\log n}$ in the algorithm of Bansal et al. one does not get a polynomial running time.

This overview raises a number of interesting questions: is it possible to replicate, or improve, the vertex
cover approximation of Bansal et al.  \cite{BCLNN17} using Sum-of-Square relaxations? A positive result
would show that, in a precise sense, $(7/8 + \varepsilon)$ approximation of Max 3SAT is ``harder'' than
$(2 - \varepsilon)$ approximation for Vertex Cover (since the former requires $\poly(\varepsilon) \cdot n$ rounds while the latter would be achievable with $n / \exp(1/\varepsilon)$ rounds). Is it possible to have a ``better than brute force''
approximation algorithm for Vertex Cover that recovers Karakostas's algorithm as a special case? Is it possible to do the same for other problems that are known to be Unique-Games-hard but not NP-hard, such as constant-factor approximation of Balanced Separator?

\subsection{Our Results}

In this work, we answer the above questions affirmatively by designing ``fast'' exponential time approximation algorithms for Vertex Cover, Uniform Sparsest Cut and related problems. For Vertex Cover, our algorithm gives $(2 - 1/O(r))$-approximation in time $\exp(n/2^{r^2})n^{O(1)}$ where $n$ is the number of vertices in the input graph and $r$ is a parameter that can be any real number at least one (and can depend on $n$). This improves upon the aforementioned recent algorithm of Bansal \etal~\cite{BCLNN17} which, for a similar approximation ratio, runs in time $\exp(n/r^r)n^{O(1)}$.  For the remaining problems, our algorithms give $O(r)$-approximation in the same running time, which improves upon a known $O(r)$-approximation algorithms with running time $\exp(n/2^r)n^{O(1)}$ that follow from~\cite{CMM10} (see the end of Section~\ref{subsec:related-works} for more details):

\begin{theorem}[Main] \label{thm:main}
For any $r > 1$, there is an $\exp(n/2^{r^2})n^{O(1)}$-time $(2 - 1/O(r))$-approximation algorithm for Vertex Cover on $n$-vertex graphs, and, there are $\exp(n/2^{r^2})n^{O(1)}$-time $O(r)$-approximation algorithms for Uniform Sparsest Cut, Balanced Separator, Min UnCut and Min 2CNF Deletion.
\end{theorem}

We remark that, when $r = C\sqrt{\log n}$ for a sufficiently large constant $C$, our algorithms coincide with the best polynomial time algorithms known for these problems~\cite{Kar09,ARV09,ACMM05}.

\subsection{Other Related Works} \label{subsec:related-works}

To prove Theorem~\ref{thm:main}, we use the Sum-of-Square relaxations of the problems and employ the conditioning framework from~\cite{BRS11,RT12} together with the main structural lemma from Arora, Rao and Vazirani's work~\cite{ARV09}. We will describe how these parts fit together in Section~\ref{sec:overview}. Before we do so, let us briefly discuss some related works not yet mentioned.

{\bf Sum-of-Square Relaxation and the Conditioning Framework.}
The Sum-of-Square (SoS) algorithm~\cite{Nes00,Par00,Las02} is a generic yet powerful meta-algorithm that can be utilized to any polynomial optimization problems. The approach has found numerous applications in both continuous and combinatorial optimization problems. Most relevant to our work is the conditioning framework developed in~\cite{BRS11,RT12}. Barak \etal~\cite{BRS11} used it to provide an algorithm for Unique Games with similar guarantee to~\cite{ABS15}, while Raghavendra and Tan~\cite{RT12} used the technique to give improved approximation algorithms for CSPs with cardinality constraints. A high-level overview of this framework is given in Sections~\ref{sec:conditioning-overview} and~\ref{sec:algo-overview}.


{\bf Approximability of Vertex Cover, Sparsest Cut and Related Problems.} All problems studied in our work are very well studied in the field of approximation algorithms and hardness of approximation. For Vertex Cover, the greedy 2-approximation algorithm has been known since the 70's (see e.g.~\cite{GJ79}). Better $(2 - \Omega(\frac{\log \log n}{\log n}))$-approximation algorithms were independently discovered in~\cite{BE85} and~\cite{MS85}. These were finally improved by Karakostas~\cite{Kar09} who used the ARV Structural Theorem to provide a $(2 - \Omega(1/\sqrt{\log n}))$-approximation for the problem. On the lower bound side, Hastad~\cite{Has01} show that $(7/6 - \varepsilon)$-approximation for Vertex Cover is NP-hard. The ratio was improved in~\cite{DS05} to 1.36. The line of works that very recently obtained the proof of the (imperfect) 2-to-1 game conjecture~\cite{KMS17,DKKMS16,DKKMS17,KMS18} also yield NP-hardness of $(\sqrt{2} - \varepsilon)$-approximate Vertex Cover as a byproduct. On the other hand, the Unique Games Conjecture (UGC)~\cite{Kho02} implies that approximating Vertex Cover to within a factor $(2 - \varepsilon)$ is NP-hard~\cite{KR08,BK09}. We remark here that only Hastad reduction (together with Moshkovitz-Raz PCP~\cite{MR10}) implies an almost exponential lower bond in terms of the running time, assuming ETH. Putting it differently, it could be the case that Vertex Cover can be approximated to within a factor $1.2$ in time say $2^{O(\sqrt{n})}$, without refuting any complexity conjectures or hypotheses mentioned here. Indeed, the question of whether a subexponential time $(2 - \varepsilon)$-approximation algorithm for Vertex Cover exists for some constant $\varepsilon > 0$ was listed as an ``interesting'' open question in~\cite{ABS15}, and it remains so even after our work.

As for (Uniform) Sparsest Cut and Balanced Separator, they were both studied by Leighton and Rao who gave $O(\log n)$-approximation algorithms for the problems~\cite{LR99}. The ratio was improved in~\cite{ARV09} to $O(\sqrt{\log n})$. In terms of hardness of approximation, these problems are \emph{not} known to be NP-hard or even UGC-hard to approximate to even just 1.001 factor. (In contrast, the non-uniform versions of both problems are hard to approximate under UGC~\cite{CKKRS06,KV15}.) Fortunately, inapproximability results of Sparsest Cut and Balanced Separator are known under stronger assumptions~\cite{Fei02,Kho06,RST12}. Specifically, Raghavendra \etal~\cite{RST12} shows that both problems are hard to approximate to any constant factor under the Small Set Expansion Hypothesis (SSEH)~\cite{RS10}. While it is not known whether SSEH follows from UGC, they are similar in many aspects, and indeed subexponential time algorithms for Unique Games~\cite{ABS15,BRS11} also work for the Small Set Expansion problem. This means, for example, that there could be an $O(1)$-approximation algorithm for both problems in subexponential time without contradicting with any of the conjectures. Whether such algorithm exists remains an intriguing open question.

Finally, both Min UnCut and Min 2CNF Deletion are shown to be approximable to within a factor of $O(\sqrt{\log n})$ in polynomial time by Agarwal \etal~\cite{ACMM05}, which improves upon previous known $O(\log n)$-approximation algorithm for Min UnCut and $O(\log n \log \log n)$-approximation algorithm for Min 2CNF Deletion by Garg \etal~\cite{GVY96} and Klein \etal~\cite{KPRT97} respectively. On the hardness side, both problems are known to be NP-hard to approximate to within $(1 + \varepsilon)$ factor for some $\varepsilon > 0$~\cite{PY91}. Furthermore, both are UGC-hard to approximate to within any constant factor~\cite{KKMO07,CKKRS06,KV15}. That is, the situations for both problems are quite similar to Sparsest Cut and Balanced Separator: it is still open whether there are subexponential time algorithms that yield $O(1)$-approximation for Min UnCut and Min 2CNF Deletion.

{\bf Fast Exponential Time Approximation Algorithms.}
As mentioned earlier, Bansal \etal~\cite{BCLNN17} recently gave a ``better than brute force'' approximation algorithm for Vertex Cover. Their technique is to first observe that we can use branch-and-bound on the high-degree vertices; once only the low-degree vertices are left, they use Halperin's (polynomial time) approximation algorithm for Vertex Cover on bounded degree graphs~\cite{Hal02} to obtain a good approximation. This approach is totally different than ours, and, given that the only way known to obtain $(2 - \Omega(1/\sqrt{\log n}))$-approximation in polynomial time is via the ARV Theorem, it is unlikely that their approach can be improved to achieve similar trade-off as ours.

\cite{BCLNN17} is not the first work that gives exponential time approximation algorithms for Vertex Cover. Prior to their work, Bourgeois \etal~\cite{BEP11} gives a $(2 - 1/O(r))$-approximation $\exp(n/r)$-time algorithm for Vertex Cover; this is indeed a certain variant of the ``limited brute force'' algorithm. Furthermore, Bansal \etal~\cite{BCLNN17} remarked in their manuscript that Williams and Yu have also independently come up with algorithms with similar guarantees to theirs, but, to the best of our knowledge, Williams and Yu's work is not yet made publicly available. 

For Sparsest Cut, Balanced Separator, Min UnCut and Min 2CNF Deletion, it is possible to derive $O(r)$-approximation algorithms that run in $\exp(n/2^r)$-time from a work of Charikar \etal~\cite{CMM10}. In particular, it was shown in~\cite{CMM10} that, for any metric space of $n$ elements, if every subset of $n/2^r$ elements can be embedded isometrically into $\ell_1$, then the whole space can be embedded into $\ell_1$ with distortion $O(r)$. Since $d$-level of Sherali-Adams (SA) relaxations for these problems ensure that every $d$-size subset of the corresponding distance metric space can be embedded isometrically into $\ell_1$, $(n/2^r)$-level of SA relaxations, which can be solved in $\exp(n/2^{\Omega(r)})$ time, ensures that the entire metric space can be embedded into $\ell_1$ with distortion $O(r)$. An algorithm with approximation ratio $O(r)$ can be derived from here, by following the corresponding polynomial time algorithm for each of the problems (\cite{LR99,Kar09,ACMM05}).

\subsection*{Organization}
In the next section, we describe the overview of our algorithms. Then, in Section~\ref{sec:prelim}, we formalize the notations and state some preliminaries. The main lemma regarding conditioned SoS solution and its structure is proved in Section~\ref{sec:condition-lemma}. This lemma is subsequently used in all our algorithms which are presented in Section~\ref{sec:alg}. We conclude our paper with several open questions in Section~\ref{sec:open}.

\section{Overview of Technique} \label{sec:overview}

Our algorithms follow the ``conditioning'' framework developed in~\cite{BRS11,RT12}. In fact, our algorithms are very simple provided the tools from this line of work, and the ARV structural theorem from~\cite{ARV09,Lee05}. To describe the ideas behind our algorithm, we will first briefly explains the ARV structural theorem and how conditioning works with Sum-of-Squares hierarchy in the next two subsections. Then, in the final subsection of this section, we describe the main insight behind our algorithms. For the ease of explaining the main ideas, we will sometimes be informal in this section; all algorithms and proofs will be formalized in the sequel.

For concreteness, we will use the $c$-Balanced Separator problem as the running example in this section. In the $c$-Balanced Separator problem, we are given a graph $G = (V, E)$ and the goal is to find a partition of $V$ into $S_0$ and $S_1 = V \setminus S_0$ that minimizes the number of edges across the cut $(S_0, S_1)$ while also ensuring that $|S_0|, |S_1| \geqs c'n$ for some constant $c' \in (0, c)$ where $n = |V|$. Note that the approximation ratio is the ratio between the number of edges cut by the solution and the optimal under the condition $|S_0|, |S_1| \geqs cn$. (That is, this is a pseudo approximation rather than a true approximation.) For the purpose of exposition, we focus only on the case where $c = 1/3$.

\subsection{The ARV Structural Theorem}

The geometric relaxation used in~\cite{ARV09} embeds each vertex $i \in V$ into a point $v_i \in \mathbb{R}^d$ such that $\|v_i\|_2 = 1$. For a partition $(S_0, S_1)$, the intended solution is $v_i = v_{\emptyset}$ if $i \in S_0$ and $v_i = -v_{\emptyset}$ otherwise, where $v_{\emptyset}$ is some unit vector. As a result, the objective function here is $\sum_{(i, j) \in E} \frac{1}{4} \|v_i - v_j\|_2^2$, and the cardinality condition $|S_0|, |S_1| \geqs n/3$ is enforced by $\sum_{i, j \in V} \|v_i - v_j\|_2^2 \geqs 8n/9$. Furthermore, Arora et al.~\cite{ARV09} also employ the triangle inequality: $\|v_i - v_j\|_2^2 \leqs \|v_i - v_k\|_2^2 + \|v_k - v_j\|_2^2$ for all $i, j, k \in V$. In other words, this relaxation can be written as follows.

\begin{align}
\text{minimize } & \sum_{(i, j) \in E} \frac{1}{4} \|v_i - v_j\|_2^2 \\ 
\text{subject to } & \sum_{i, j \in [n]} \|v_i - v_j\|_2^2 \geqs 8n/9 \label{eq:spread-cond} \\
& \|v_i\|_2^2 = 1 & \forall i \in V \label{eq:unit-cond} \\
& \|v_i - v_j\|_2^2 \leqs \|v_i - v_k\|_2^2 + \|v_k - v_j\|_2^2 & \forall i, j, k \in V \label{eq:triangle-inq-cond}
\end{align}

Note here that the above relaxation can be phrased as a semidefinite program and hence can be solved to arbitrarily accuracy in polynomial time. The key insight shown by Arora \etal\ is that, given a solution $\{v_i\}_{i \in V}$ to the above problem, one can find two sets of vertices $T, T'$ that are $\Omega(1/\sqrt{\log n})$ apart from each other, as stated below. Note that this version is in fact from~\cite{Lee05}; the original theorem of~\cite{ARV09} has a worst parameter with $\Delta = \Omega((\log n)^{2/3})$.

\begin{theorem}[ARV Structural Theorem~\cite{ARV09,Lee05}]
Let $\{v_i\}_{i \in V}$ be any vectors in $\mathbb{R}^d$ satisfying~\eqref{eq:spread-cond},~\eqref{eq:unit-cond},~\eqref{eq:triangle-inq-cond}. There exist disjoint sets $T, T' \subseteq V$ each of size $\Omega(n)$ such that, for every $i \in T$ and $j \in T'$, $\|v_i - v_j\|_2^2 \geqs \Delta = \Omega(1/\sqrt{\log n})$. Moreover, such sets can be found in randomized polynomial time.
\end{theorem}

It should be noted that, given the above theorem, it is easy to arrive at the $\Omega(1/\sqrt{\log n})$-approximation algorithm for balanced separator. In particular, we can pick a number $\theta$ uniformly at random from $[0, \Delta)$ and then output $S_0 = \{i \in V \mid \exists j \in T, \|v_i - v_j\|_2^2 \leqs \theta\}$ and $S_1 = V \setminus S_0$. It is easy to check that the probability that each edge $(i, j) \in E$ is cut is at most $\|v_i - v_j\|_2^2/\Delta = O(\sqrt{\log n} \cdot \|v_i - v_j\|_2^2)$. Moreover, we have $|S_0| \geqs |T| \geqs \Omega(n)$ and $|S_1| \geqs |T'| \geqs \Omega(n)$, meaning that we have arrived at an $O(\sqrt{\log n})$-approximate solution for Balanced Separator.

An interesting aspect of the proof of~\cite{Lee05} is that the bound on $\Delta$ can be improved if the solution $\{v_i\}_{i \in V}$ is ``hollow'' in the following sense: for every $i \in V$, the ball of radius\footnote{Here 0.1 can be changed to arbitrary positive constant; we only use it to avoid introducing additional parameters.} 0.1 around $i$ contains few other vectors $v_j$'s. In particular, if there are only $m$ such $v_j$'s, then $\Delta$ can be made $\Omega(1/\sqrt{\log m})$, instead of $\Omega(1/\sqrt{\log n})$ in the above version. We will indeed use this more fine-grained version (in a black-box manner) in our algorithms. To the best of our knowledge, this version of the theorem has not yet been used in other applications of the ARV Structural Theorem.

\begin{theorem}[Refined ARV Structural Theorem~\cite{ARV09,Lee05}] \label{thm:arv-refined}
Let $\{v_i\}_{i \in V}$ be any vectors in $\mathbb{R}^d$ satisfying~\eqref{eq:spread-cond},~\eqref{eq:unit-cond},~\eqref{eq:triangle-inq-cond}. Moreover, let $m = \max_{i \in V} |\{j \in V \mid \|v_i - v_j\|_2^2 \leqs 0.01\}|$. There exist disjoint sets $T, T' \subseteq V$ each of size $\Omega(n)$ such that, for every $i \in T$ and $j \in T'$, $\|v_i - v_j\|_2^2 \geqs \Delta = \Omega(1/\sqrt{\log m})$. Moreover, such sets can be found in randomized polynomial time.
\end{theorem}

\subsection{Conditioning in Sum-of-Square Hierarchies} \label{sec:conditioning-overview}

Another crucial tool used in our algorithm is Sum-of-Square hierarchy and the conditioning technique developed in~\cite{BRS11,RT12}. Perhaps the most natural interpretation of the sum-of-square solution with respect to the conditioning operation is to view the solution as local distributions. One can think of a degree-$d$ sum-of-square solution for Balanced Separator as a collection of local distributions $\mu_S$ over $\{0, 1\}^S$ for subsets of vertices $S \subseteq V$ of sizes at most $d$ that satisfies certain consistency and positive semi-definiteness conditions, and additional linear constraints corresponding to $|S_0|, |S_1| \geqs n/3$ and the triangle inequalities. More specifically, for every $U \subseteq V$ and every $\phi: U \to \{0, 1\}$, the degree-$d$ sum-of-squares solution gives us $\Pr_{\mu_U}[\forall_{j \in U}, j \in S_{\phi(j)}]$ which is a number between zero and one. The consistency constraints ensures that these distributions are locally consistent; that is, for every $U' \subseteq U \in \{0, 1\}$, the marginal distribution of $\mu_U$ on $U'$ is equal to $\mu_{U'}$. We remark here that, for Balanced Separator and other problems considered in this work, a solution to the degree-$d$ SoS relaxation for them can be found in time $\binom{n}{d}^{O(1)} = O(n/d)^{O(d)}$.

This consistency constraint on these local distributions allow us to define conditioning on local distributions in the same ways as typical conditional distributions. For instance, we can condition on the event $i \in S_0$ if $\Pr_{\mu_i}[i \in S_0] \ne 0$; this results in local distributions $\{\mu'_U\}_{U \subseteq V, |U| \leqs d - 1}$ where $\mu'_U$ is the conditional distribution of $\mu_{U \cup \{i\}}$ on the event $i \in S_0$. In other words, for all $\phi: U \to \{0, 1\}$,
\begin{align*}
\Pr_{\mu'_U}[\forall j \in U, j \in S_{\phi(j)}] =
\frac{\Pr_{\mu_{U \cup \{i\}}}\left[i \in S_0 \wedge \left(\forall j \in U, j \in S_{\phi(j)}\right)\right]}{\Pr_{\mu_i}[i \in S_0]}.
\end{align*}
Notice that the local distributions are now on subsets of at most $d - 1$ vertices instead of on subsets of at most $d$ vertices. In other words, the conditioned solution is a degree-$(d - 1)$ solution.

As for the semi-definiteness constraint, it suffices for the purpose of this discussion to think about only the degree-2 solution case. For this case, the semi-definiteness constraint in fact yields unit vectors $v_{\emptyset}, \{v_{j}\}_{j \in V}$ such that
\begin{align*}
\Pr_{\mu_i}[i \in S_0] &= \frac{1 + \left<v_{\emptyset}, v_i\right>}{2} &\forall i \in V, \\
\Pr_{\mu_{\{i, j\}}}[i, j \in S_0] &= \frac{1 + \left<v_{\emptyset}, v_i\right> + \left<v_{\emptyset}, v_j\right> + \left<v_i, v_j\right>}{4} &\forall i, j \in V.
\end{align*}
It is useful to also note that the probability that $i, j$ are on different side of the cut is exactly equal to $\frac{1}{4}\|v_i - v_j\|_2^2$; this is just because
\begin{align} \label{eq:diff-prob}
\Pr_{\mu_{\{i, j\}}}[Y_i \ne Y_j] &= \Pr_{\mu_i}[i \in S_0] + \Pr_{\mu_j}[j \in S_0] - 2\Pr_{\mu_{\{i, j\}}}[i \in S_0 \wedge j \in S_0] = \frac{1 - \left<v_i, v_j\right>}{2} = \frac{1}{4} \|v_i - v_j\|_2^2,
\end{align}
where $Y_i, Y_j$ are boolean random variables such that $i \in S_{Y_i}$ and $j \in S_{Y_j}$.

Finally, we note that the constraints for $|S_0|, |S_1| \geqs n/3$ and the triangle inequalities are those that, when written in vector forms, translate to inequalities~\eqref{eq:spread-cond} and ~\eqref{eq:triangle-inq-cond} from the ARV relaxation.

\subsection{Our Algorithms: Combining Conditioning and the ARV Theorem} \label{sec:algo-overview}

The conditioning framework initiated in~\cite{BRS11,RT12} (and subsequently used in~\cite{ABG13,YZ14,MR16}) typically proceeds as follows: solve for a solution to a degree-$d$ Sum-of-Square relaxation of the problem for a carefully chosen value of $d$, use (less than $d$) conditionings to make a solution into an ``easy-to-round'' degree-$O(1)$ solution, and finally round such a solution.

To try to apply this with the Balanced Separator problem, we first have to understand what are the ``easy-to-round'' solutions for the ARV relaxation. In this regards, first observe that, due to the more refined version of the ARV Theorem (Theorem~\ref{thm:arv-refined}), the approximation ratio is actually $O(\sqrt{\log m})$ which can be much better than $O(\sqrt{\log n})$. In particular, if $m \leqs 2^{O(r^2)}$, this already yields the desired $O(r)$-approximation algorithm. This will be one of the ``easy-to-round'' situations. Observe also that we can in fact relax the requirement even further: it suffices if $|\{j \in V \mid \|v_i - v_j\|_2^2 \leqs 0.01\}| \leqs m$ holds for a constant fraction of vertices $i \in V$. This is because we can apply Theorem~\ref{thm:arv-refined} on only the set of such $i$'s which would still result in well-separated set of size $\Omega(n)$. Recall also that from~\eqref{eq:diff-prob} the condition $\|v_i - v_j\|_2^2 \leqs 0.01$ is equivalent to $\Pr_{\mu_{\{i, j\}}}[Y_i \ne Y_j] \leqs 0.04$.

Another type of easy-to-round situation is when, for most (i.e. $0.9n$) of $i \in V$, $\Pr_{\mu_i}[i \in S_0] \notin [0.2, 0.8]$. In this latter scenario, we can simply find a pair of large well-separated sets $(T, T')$ by just letting $T = \{i \in V \mid \Pr_{\mu_i}[i \in S_0] < 0.2\}$ and $T' = \{j \in V \mid \Pr_{\mu_j}[j \in S_0] > 0.8\}$. It is not hard to argue that both $T, T'$ are at least $\Omega(n)$ and that, for every $i \in T$ and $j \in T'$, $\|v_i - v_j\|_2^2$ is at least $0.6$.

To recap, it suffices for us to condition degree-$d$ solution so that we end up in one of the following two ``easy-to-round'' cases in order to get $O(r)$ approximation algorithm for the problem.
\begin{enumerate}
\item For at least $n/100$ vertices $i \in V$, we have $|\{j \in V \mid \Pr_{\mu_{\{i, j\}}}[Y_i \ne Y_j] \leqs 0.04\}| \leqs 2^{O(r^2)}$.
\item For at least $9n/10$ vertices $i \in V$, we have $\Pr_{\mu_i}[i \in S_0] \notin [0.2, 0.8]$.
\end{enumerate}
Here we will pick our $d$ to be $n/2^{r^2}$; the running time needed to solve for such a solution is indeed $O(n/d)^{O(d)} = \exp(n/2^{O(r^2)})n^{O(1)}$ as claimed. Now, suppose that we have a degree-$d$ solution that does not belong to any of the two easy-to-round cases as stated above. This means that there must be $i \in V$ such that $\Pr_{\mu_i}[i \in S_0] \notin [0.2, 0.8]$ and that $|\{j \in V \mid \Pr_{\mu_{\{i, j\}}}[Y_i \ne Y_j] \leqs 0.04\}| > 2^{O(r^2)}$. For simplicity, let us also assume for now that $\Pr_{\mu_i}[i \in S_0] = 0.5$. We will condition on the event $i \in S_0$; let the local distributions after conditioning be $\{\mu'_U\}_{U \subseteq V, |U| \leqs d - 1}$. Consider each $j \in V$ such that $\Pr_{\mu_{\{i, j\}}}[Y_i \ne Y_j] \leqs 0.04$. Observe first that, before the conditioning, we have 
\begin{align*}
\Pr_{\mu_j}[j \in S_0] \geqs \Pr_{\mu_i}[i \in S_0] - \Pr_{\mu_\{i, j\}}[Y_i \ne Y_j] > 0.4
\end{align*}
and
\begin{align*}
\Pr_{\mu_j}[j \in S_0] \geqs \Pr_{\mu_i}[i \in S_0] + \Pr_{\mu_\{i, j\}}[Y_i \ne Y_j] < 0.6.
\end{align*}
On the other hand, after the conditioning, we have
\begin{align*}
\Pr_{\mu'_j}[j \in S_0] &= \frac{\Pr_{\mu_{\{i, j\}}}[i \in S_0, j \in S_0]}{\Pr_{\mu_i}[i \in S_0]} \\
&= \frac{\Pr_{\mu_i}[i \in S_0] - \Pr_{\mu_{\{i, j\}}}[i \in S_0, j \in S_1]}{\Pr_{\mu_i}[i \in S_0]} \\
&\geqs \frac{\Pr_{\mu_i}[i \in S_0] - \Pr_{\mu_{\{i, j\}}}[Y_i \ne Y_j]}{\Pr_{\mu_i}[i \in S_0]} \\
&\geqs 1 - 0.04/0.5 \\
&> 0.9.
\end{align*}
Thus, this conditioning makes at least $2^{r^2}$ vertices $j$'s such that $\Pr_{\mu_j}[j \in S_0] \notin [0.2, 0.8]$ beforehand satisfy $\Pr_{\mu'_j}[j \in S_0] \in [0.2, 0.8]$ afterwards. If we ignore how conditioning affects the remaining variables for now, this means that, after $n/2^{r^2}$ such conditioning all vertices $j \in V$ must have $\Pr_{\mu_j}[j \in S_0] \in [0.2, 0.8]$. Hence, we have arrived at an ``easy-to-round'' solution and we are done! The effect to the other variables that we ignored can easily be taken into account via a simple potential function argument and by considering conditioning on both $i \in S_0$ and $i \in S_1$; this part of the argument can be found in Section~\ref{sec:condition-lemma}. This concludes the overview of our algorithm.


\section{Preliminaries} \label{sec:prelim}

\subsection{Sum-of-Square Hierarchy, Pseudo-Distribution, and Conditioning}

We define several notations regarding the Sum-of-Square (SoS) Hierarchy; these notations are based mainly on~\cite{BBHKSZ12,OZ13}. We will only state preliminaries necessary for our algorithms. We recommend interested readers to refer to~\cite{OZ13,BS14} for a more thorough survey on SoS.

We use $\mathbb{R}_d[X_1, \dots, X_n]$ to denote the set of all polynomials on $X_1, \dots, X_n$ of total degree at most $d$. First, we define the notion of \emph{pseudo-expectation}, which represents solutions to SoS Hierarchy:

\begin{definition}[Pseudo-Expectation]
A \emph{degree-$d$ pseudo-expectation} is a linear operator $\pE: \mathbb{R}_d[X_1, \dots, X_n] \to \mathbb{R}$ that satisfies the following:
\begin{itemize}
\item (Normalization) $\pE[1] = 1$.
\item (Linearity) For any $p \in \mathbb{R}_d[X_1, \dots, X_n]$ and $q \in \mathbb{R}_d[X_1, \dots, X_n]$, $\pE[p + q] = \pE[p] + \pE[q]$.
\item (Positivity) For any $p \in \mathbb{R}_{\lfloor d/2\rfloor}[X_1, \dots, X_n]$, $\pE[p^2] \geqs 0$.
\end{itemize}
Furthermore, $\pE$ is said to be \emph{boolean} if $\pE[(X_i^2 - 1)p] = 0$ for all $p \in \mathbb{R}_{d - 2}[X_1, \dots, X_n]$.
\end{definition}

Observe that, while $\pE$ is a function over infinite domain, $\pE$ has a succinct representation: due to its linearity, it suffices to specify the values of all monomials of total degree at most $d$ and there are only $n^{O(d)}$ such monomials. Furthermore, for boolean $\pE$, we can save even further since it suffices to specify only products of at most $d$ different variables. There are only $O(n/d)^{O(d)}$ such terms. From now on, we will only consider boolean pseudo-expectations. Note also that we use $X_i$ as $\pm 1$ variables instead of $0, 1$ variable as used in the proof overview. (Specifically, in the language of the proof overview section, $\Pr_{\mu_i}[i \in S_0]$ is now equal to $\pE[(1 - X_i)/2]$.)

\begin{definition}
A \emph{system of polynomial constraints} $(\cP, \cQ)$ consists of the set of equality constraints $\cP = \{p_i = 0\}_{i \in |\cP|}$ and the set of inequality constraints $\cQ = \{q_j \geqs 0\}_{j \in |\cQ|}$, where all $p_i$ and $q_j$ are polynomials over $X_1, \dots, X_n$. We denote the degree of $(\cP, \cQ)$ by $\deg(\cP, \cQ) := \max\{\deg(p_i), \deg(q_j)\}_{i \in |\cP|, j \in |\cQ|}$ where $\deg(p)$ denote the (total) degree of polynomial $p$.

For every $S \subseteq [n]$, we use $X_S$ to denote the monomial $\prod_{i \in S}X_i$. Furthermore, for every $S \subseteq [n]$ and every $\phi: S \to \{\pm 1\}$, let $X_{\phi}$ be the polynomial $\prod_{i \in S} \left(1 + \phi(i)X_i\right)$.
A boolean degree-$d$ pseudo-expectation $\pE: \mathbb{R}_d[X_1, \dots, X_n] \to \mathbb{R}$ is said to \emph{satisfy} a system of polynomial constraints $(\cP, \cQ)$ if the following conditions hold:
\begin{itemize}
\item For all $p \in \cP$ and all $S \subseteq [n]$ such that $|S| \leqs d - \deg(p)$, we have $\pE[X_S p] = 0$.
\item For all $q \in \cQ$, all $S \subseteq [n]$ such that $|S| \leqs d - \deg(q)$ and all $\phi: S \to \{\pm 1\}$, we have $\pE[X_{\phi} q] \geqs 0$.
\end{itemize}
\end{definition}

Note that there are only $O(n/d)^{O(d)}$ equalities and inequalities generated above; indeed all degree-$d$ SoS relaxations considered in our work can be solved in time $O(n/d)^{O(d)}$ since it can be expressed as a semidefinite program\footnote{It has been recently pointed out by O'Donnell~\cite{O17} that the fact that SoS can be written as small SDP is not sufficient to conclude the bound on the running time. However, this is not an issue for us since we are working with the primal solutions (as opposed to sum-of-square certificates) and we can tolerate small errors in each of the equalities and inequalities. In particular, the ellipsoid algorithm can find, in time polynomial of the size of the program, a solution where the error in each inequality is at most say $2^{-n^{100}}$, and this suffices for all of our algorithms.} of size $O(n/d)^{O(d)}$.

\begin{definition}[Conditioning] \label{def:cond}
Let $\pE: \mathbb{R}_d[X_1, \dots, X_n] \to \mathbb{R}$ be any boolean degree-$d$ pseudo-expectation for some $d > 2$. For any $b \in \{\pm 1\}$ such that $\pE[X_i] \ne -b$, we denote the conditional pseudo-expectation of $\pE$ on $X_i = b$ by $\pE|_{X_i = b}: \mathbb{R}_{d - 1}[X_1, \dots, X_n] \to \mathbb{R}$ where
\begin{align*}
\pE|_{X_i = b}[p] = \frac{\pE[p(1 + bX_i)]}{\pE[1 + bX_i]}
\end{align*}
for all $p \in \mathbb{R}_{d - 1}[X_1, \dots, X_n]$.
\end{definition}

The proposition below is simple to check, using the identity $(1 + bX_i) = \frac{1}{2}(1 + bX_i)^2$.

\begin{proposition}
Let $\pE, b, \pE|_{X_i = b}$ be as in Definition~\ref{def:cond}. If $\pE$ satisfies a system of polynomial constraints $(\cP, \cQ)$, then $\pE|_{X_i = b}$ also satisfies the system $(\cP, \cQ)$.
\end{proposition}

\subsection{ARV Structural Theorems}

Having defined appropriate notations for SoS, we now move on to another crucial preliminary: the ARV Structural Theorem. It will be useful to state the theorem both in terms of metrics and in terms of pseudo-expectation. Let us start by definitions of several notations for metrics.

\begin{definition}[Metric-Related Notations] \label{def:metric}
A \emph{metric} $d$ on $X$ is a distance function $d: X \times X \to \mathbb{R}_{\geqs 0}$ that satisfies\footnote{Here we do not require ``identity of indiscernibles
'' (i.e. $d(x, y) = 0$ if and only if $x = y$), which is sometimes an axiom for metrics in literature. Without such a requirement, $d$ is sometimes referred to as a \emph{pseudometric}.} (1) $d(x, x) = 0$, (2) symmetry $d(x, y) = d(y, x)$ and (3) triangle inequality $d(x, z) \leqs d(x, y) + d(y, z)$, for all $x, y, z \in X$. We use the following notations throughout this work:
\begin{itemize}
\item For $x \in X$ and $S, T \subseteq X$, $d(x, S) := \min_{y \in S} d(x, y)$ and $d(S, T) := \min_{y \in S} d(y, T)$. 
\item We say that $S, T$ are \emph{$\Delta$-separated} iff $d(S, T) \geqs \Delta$. 
\item The \emph{diameter} of a metric space $(X, d)$ denoted by $\diam(X, d)$ is $\max_{x, y \in X} d(x, y)$.
\item We say that $(X, d)$ is $\alpha$-spread if $\sum_{x, y \in X} d(x, y) \geqs \alpha|X|^2$.
\item An \emph{(open) ball of radius $r$ around $x$} denoted by $\cB_d(x, r)$ is defined as $\{y \in X \mid d(x, y) < r\}$. 
\item A metric space $(X, d)$ is said to be \emph{($r$, $m$)-hollow} if $|\cB_d(x, r)| \leqs m$ for all $x \in X$.
\end{itemize}
\end{definition}

\begin{definition}[Negative Type Metric]
A metric space $(X, d)$ is said to be of \emph{negative type} if $\sqrt{d}$ is Euclidean. That is, there exists $f: X \to \mathbb{R}^q$ such that $\|f(x) - f(y)\|_2^2 = d(x, y)$ for all $x, y \in X$.
\end{definition}

The ARV Theorem states that, in any negative type metric space $(X, d)$ that is $\Omega(\diam(d))$-spread and $(\Omega(\diam(d)), m)$-hollow, there exists two large subsets that are $\Omega\left(\frac{\diam(d)}{\sqrt{\log m}}\right)$-separated:

\begin{theorem}[ARV Structural Theorem - Metric Formulation~\cite{ARV09,Lee05}] \label{thm:arv-metric}
Let $\alpha, r > 0$ be any positive real number and $m \in \mathbb{N}$ be any positive integer. For any negative type metric space $(X, d)$ with $\diam(d) \leqs 1$ that is $\alpha$-spread and $(r, m)$-hollow, there exist disjoints subsets $T, T' \subseteq X$ each of size $\Omega_{\alpha, r}(|X|)$ such that $d(T, T') \geqs \Omega_{\alpha, r}(1/\sqrt{\log m})$. Moreover, these sets can be found in randomized polynomial time.
\end{theorem}

We remark that the quantitative bound $\Delta = \Omega_{\alpha, \beta}(1/\sqrt{\log m})$ comes from Lee's version of the theorem~\cite{Lee05} whereas the original version only have $\Delta = \Omega_{\alpha, \beta}(1/(\log m)^{2/3})$. We also note that even Lee's version of the theorem is not stated exactly in the above form; in particular, he only states the theorem with $m = |X|$, for which the Hollowness condition is trivial. We will neither retread his whole argument nor define all notations from his work here, but we would like to point out that it is simple to see that his proof implies the version that we use as well. Specifically, the inductive hypothesis in the proof of Lemma 4.2 of~\cite{Lee05} implies that when the procedure fails (with constant probability) to find $T, T'$ that are separated by $\Delta = C/\sqrt{\log m}$ where $C = C(\alpha, r)$ is sufficiently large, then there exists $S \subseteq X$ that is $(100\sqrt{\log m/r}, 0.1, \sqrt{2r})$-covered by $X$. Lemma 4.1 of~\cite{Lee05} then implies that, for each $x \in S$, we must have $|\cB(x, r)| > m$.

As we are using the ARV Theorem in conjunction with the SoS conditioning framework, it is useful to also state the theorem in SoS-based notations. To do so, let us first state the following fact, which can be easily seen via the fact that the moment matrix (with $(i, j)$-entry equal to $\pE[X_iX_j]$) is positive semidefinite and thus is a Gram matrix for some set of vectors:

\begin{proposition} \label{prop:pE-metric}
Let $\pE: \mathbb{R}_2[X_1, \dots, X_n] \to \mathbb{R}$ be any degree-2 pseudo-expectation that satisfies the triangle inequality $\pE[(X_i - X_j)^2] \leqs \pE[(X_i - X_k)^2] + \pE[(X_k - X_j)^2]$ for all $i, j, k \in [n]$. Define $d_{\pE}: [n] \times [n] \to \mathbb{R}_{\geqs 0}$ by $d_{\pE}(i, j) = \pE[(X_i - X_j)^2]$. Then, $([n], d_{\pE})$ is a negative type metric space.
\end{proposition} 

When it is clear which pseudo-expectation we are referring to, we may drop the subscript from $d_{\pE}$ and simply write $d$. Further, we use all metric terminologies with $\pE$ in the natural manner; for instance, we say that $S, T \subseteq [n]$ are $\Delta$-separated if $d_{\pE}(S, T) \geqs \Delta$.

Theorem~\ref{thm:arv-metric} can now be restated in pseudo-expectation notations as follows.

\begin{theorem}[ARV Structural Theorem - SoS Formulation~\cite{ARV09,Lee05}] \label{thm:arv-sos}
For any $\alpha, \beta > 0$ and $m \in \N$, let $\pE: \mathbb{R}_2[X_1, \dots, X_n] \to \mathbb{R}$ be any degree-2 pseudo-expectation such that the following conditions hold:
\begin{itemize}
\item (Boolean) For every $i \in [n]$, $\pE[X_i^2] = 1$. 
\item (Triangle Inequality) For every $i, j, k \in [n]$, $\pE[(X_i - X_j)^2] \leqs \pE[(X_i - X_k)^2] + \pE[(X_k - X_j)^2]$. 
\item (Balance) $\sum_{i, j \in [n]} \pE[(X_i - X_j)^2] \geqs \alpha n^2.$
\item (Hollowness) For all $i \in [n]$, $|\{j \in [n] \mid \pE[X_iX_j] > 1 - \beta\}| \leqs m$.
\end{itemize}
Then, there exists a randomized polynomial time algorithm that, with probability 2/3, produces disjoint subsets $T, T' \subseteq [n]$ each of size at least $\Omega_{\alpha, \beta}(n)$ such that $T, T'$ are $\Delta$-separated for $\Delta = \Omega_{\alpha, \beta}(1/\sqrt{\log m})$.
\end{theorem}

Notice that, for boolean $\pE$, $\pE[X_iX_j] = 1 - \pE[(X_i - X_j)^2]/2 = 1 - d_{\pE}(i, j)/2$. This means that $\{j \in [n] \mid \pE[X_iX_j] > 1 - \beta\}$ is simply $\cB_{d_{\pE}}(i, 2\beta)$. Another point to notice is that the metric $d_{\pE}$ can have $\diam(d_{\pE})$ as large as 4, instead of 1 required in Theorem~\ref{thm:arv-metric}, but this poses no issue since we can scale all distances down by a factor of 4.

We also need a slight variant of the theorem that does not require the balanceness constraint; such variant appears in~\cite{Kar09,ACMM05}. It is proved via the ``antipodal trick'' where, for every $i \in [n]$, one also add an additional variable $X_{-i}$ and add the constraint $\pE[X_i + X_{-i}] = 0$ to the system. Applying the above lemma together with an observation that the procedure to creates a set from~\cite{ARV09} can be modified so that $i \in T$ iff $-i \in T'$ gives the following:

\begin{corollary}[ARV Structural Theorem for Antipodal Vectors~\cite{Kar09}] \label{cor:arv-antipodal}
Let $\pE: \mathbb{R}_2[X_1, \dots, X_n] \to \mathbb{R}$ be any degree-2 pseudo-expectation that satisfies the following conditions for any $\beta > 0$ and $m \in \N$:
\begin{itemize}
\item (Boolean) For every $i \in [n]$, $\pE[X_i^2] = 1$. 
\item (Triangle Inequality) For every $i, j, k \in [n]$, 
\begin{align*}
\pE[(X_i - X_j)^2] \leqs \pE[(X_i - X_k)^2] + \pE[(X_k - X_j)^2], \\
\pE[(X_i - X_j)^2] \leqs \pE[(X_i + X_k)^2] + \pE[(X_k + X_j)^2], \\
\pE[(X_i + X_j)^2] \leqs \pE[(X_i - X_k)^2] + \pE[(X_k + X_j)^2], \\
\pE[(X_i + X_j)^2] \leqs \pE[(X_i + X_k)^2] + \pE[(X_k - X_j)^2].
\end{align*}
\item (Hollowness) For all $i \in [n]$, $|\{j \in [n] \mid |\pE[X_iX_j]| > 1 - \beta\}| \leqs m$.
\end{itemize}
Then, there exists a randomized polynomial time algorithm that, with probability 2/3, produces disjoint subsets $T, T' \subseteq [n]$ such that $|T| + |T'| \geqs \Omega_{\beta}(n)$ and, for every $i, i' \in T$ and $j, j' \in T'$, we have $\pE[(X_i - X_j)^2], \pE[(X_i + X_{i'})^2], \pE[(X_j + X_{j'})^2] \geqs \Omega_{\beta}(1/\sqrt{\log m})$.
\end{corollary}

\subsection{The Problems}

The following are the list of problems we consider in this work.

{\bf Vertex Cover.} A subset $S \subseteq V$ of vertices is said to be a vertex cover of $G = (V, E)$ if, for every edge $\{u, v\} \in E$, $S$ contains at least one of $u$ or $v$. The goal of the vertex cover problem is to find a vertex cover of minimum size.

{\bf Sparsest Cut.} Given a graph $G = (V, E)$. The (edge) expansion of $S \subseteq V$ is defined as $\Phi_G(S) = \frac{|E(S, V \setminus S)|}{\min\{|S|, |V \setminus S|\}}$, where $E(S, V \setminus S)$ denote the set of edges across the cut $(S, V \setminus S)$. In the uniform sparsest cut problem, we are asked to find a subset of vertices $S$ that minimizes $\Phi_G(S)$.

{\bf Balanced Separator.} In the Balanced Separator problem, the input is a graph $G = (V, E)$ and the goal is to find a partition of $V$ into $S_0, S_1$ with $S_0, S_1 \geqs c'|V|$ for some constant $c' > 0$ such that $\Phi_G(S_0)$ is minimized. Note that the approximation ratio is with respect to the minimum $\Phi_G(S_0)$ for all partition $S_0, S_1$ such that $|S_0|, |S_1| \geqs c|V|$ where $c$ is some constant greater than $c'$. In other words, the algorithm is a pseudo (aka bi-criteria) approximation; this is also the notion used in~\cite{LR99,ARV09}.

For simplicity, we only consider the case where $c = 1/3$ in this work; it is easy to see that the algorithm provided below can be extended to work for any constant $c \in (0, 1)$.

{\bf Minimum UnCut.} Given a graph $G = (V, E)$, the Minimum UnCut problem asks for a subset $S \subseteq V$ of vertices that minimizes the number of edges that do \emph{not} cross the cut $(S, V \setminus S)$.

{\bf Minimum 2CNF Deletion.} In this problem, we are given a 2CNF formula and the goal is to find a minimum number of clauses such that, when they are removed, the formula becomes satisfiable. Here we use $n$ to denote the number of variables in the input formula.

\section{Conditioning Yields Easy-To-Round Solution} \label{sec:condition-lemma}

The main result of this section is the following lemma on structure of conditioned solution:

\begin{lemma} \label{lem:main-conditioning}
Let $\tau, \gamma$ be any positive real numbers such that $\tau^2 < \gamma < 1$. Given a boolean degree-$d$ pseudo-expectation $\pE: \mathbb{R}_{d}[X_1, \dots, X_n] \to \mathbb{R}$ for a system $(\cP, \cQ)$ and an integer $\ell < d$, we can, in time $O(n/d)^{O(d)}$, find a boolean degree-$(d-\ell)$ pseudo-expectation $\pE': \mathbb{R}_{d - \ell}[X_1, \dots, X_n] \to \mathbb{R}$ for the system $(\cP, \cQ)$ such that the following condition holds: 
\begin{itemize}
\item Let $V_{(-\tau, \tau)} := \{i \in [n] \mid \pE'[X_i] \in (-\tau, \tau)\}$ denote the set of indices of variables whose pseudo-expectation lies in $(-\tau, \tau)$ and, for each $i \in [n]$, let $C_{\gamma}(i) := \{j \in [n] \mid \pE'[X_iX_j] \in [-\gamma, \gamma]\}$ denote the set of all indices $j$'s such that $\pE'[X_iX_j]$ lies in $[-\gamma, \gamma]$. Then, for all $i \in V_{(-\tau, \tau)}$, we have $$|V_{(-\tau, \tau)} \setminus C_{\gamma}(i)| \leqs \frac{n}{\ell(\gamma - \tau^2)^2}.$$
\end{itemize}
\end{lemma}

In other words, the lemma says that, when $d$ is sufficiently large, we can condition so that we arrive at a pseudo-expectation with the hollowness condition if we restrict ourselves to $V_{(-\tau, \tau)}$. Note here that, outside of $V_{(-\tau, \tau)}$, this hollowness condition does not necessarily hold. For instance, it could be that after conditioning all variables be come integral (i.e. $\pE[X_i] \in \{\pm 1\}$). However, this is the second ``easy-to-round'' case for ARV theorem, so this does not pose a problem for us.

The proof of Lemma~\ref{lem:main-conditioning} will be based on a potential function argument. In particular, the potential function we use is $\Phi(\pE) = \sum_{i \in [n]} \pE[X_i]^2$. The main idea is that, as long as there is a ``bad'' $i \in [n]$ that violates the condition states in the lemma, we will be able to finding a conditioning that significantly increases $\Phi$. However, $\Phi$ is always at most $n$, meaning that this cannot happens too many times and, thus, we must at some point arrive at a pseudo-distribution with no bad $i$.

To facilitate our proof, let us prove a simple identity regarding the potential change for a single variable after conditioning:

\begin{proposition} \label{prop:diff-cond}
Let $\pE: \mathbb{R}_{d}[X_1, \dots, X_n] \to \mathbb{R}$ be any degree-$d$ pseudo-expectation for some $d > 2$ and let $i \in [n]$ be such that $\pE[X_i] \ne -1, 1$. Then, for any $j \in [n]$, we have
\begin{align*}
\left(\frac{1 - \pE[X_i]}{2}\right) \left(\pE|_{X_i=-1}[X_j]\right)^2 + \left(\frac{1 + \pE[X_i]}{2}\right) \left(\pE|_{X_i=1}[X_j]\right)^2 - \pE[X_j]^2 = \frac{(\pE[X_i X_j] - \pE[X_i]\pE[X_j])^2}{1 - \pE[X_i]^2}.
\end{align*}
\end{proposition}

\begin{proof}
For succinctness, let $a = (1 - \pE[X_i])/2, b = \pE|_{X_i = 1}[X_j]$ and $c = \pE|_{X_i = -1}[X_j]$. Observe that, from definition of conditioning, we have
\begin{align*}
ab + (1 - a)c = \pE[(1 - X_i)X_j/2] + \pE[(1 + X_i)X_j/2] = \pE[X_j].
\end{align*}
Hence, the left hand side term of the equation in the proposition statement can be rewritten as 
\begin{align} \label{eq:diff-succ}
ab^2 + (1 - a)c^2 - (ab + (1 - a)c)^2
&= a(1 - a)(b - c)^2.
\end{align}
Let $\mu_i = \pE[X_i]$. Now, observe that $b - c$ is simply
\begin{align*}
\pE|_{X_i = -1}[X_j] - \pE|_{X_i = 1}[X_j] &= \frac{\pE[(1 - X_i)X_j]}{1 - \mu_i} - \frac{\pE[(1 + X_i)X_j]}{1 + \mu_i} \\
&= \frac{\pE[(1 + \mu_i)(1 - X_i)X_j - (1 - \mu_i)(1 + X_i)X_j]}{1 - \mu_i^2} \\
&= \frac{\pE[2(\mu_i - X_i)X_j]}{1 - \mu_i^2} \\
&= \frac{2(\mu_i\pE[X_j] - \pE[X_iX_j])}{1 - \mu_i^2}.
\end{align*}
Plugging the above equality back into~\eqref{eq:diff-succ} yields the desired identity.
\end{proof}

With the above lemma ready, we now proceed to the proof of Lemma~\ref{lem:main-conditioning}. Before we do so, let us also note that our choice of potential function $\pE[X_i]^2$ is not of particular importance; indeed, there are many other potential functions that work, such as the entropy of $X_i$.

\begin{proof}[Proof of Lemma~\ref{lem:main-conditioning}]
We describe an algorithm below that finds $\pE'$ by iteratively conditioning the pseudo-distribution on the variable $X_i$ that violates the condition.
\begin{enumerate}
\item Let $\pE_0 = \pE$
\item For $t = 1, \dots, \ell$, execute the following steps.
\begin{enumerate}
\item Let $V_{(-\tau, \tau)}^{t - 1} := \{i \in [n] \mid \pE_{t - 1}[X_i] \in (-\tau, \tau)\}$. \\Moreover, for each $i \in [n]$, let $C^{t - 1}_{\gamma}(i) := \{j \in [n] \mid \pE_{t - 1}[X_iX_j] \in [-\gamma, \gamma]\}$.
\item If $|V_{(-\tau, \tau)}^{t - 1} \setminus C^{t - 1}_{\gamma}(i)| \leqs \frac{n}{\ell(\gamma - \tau^2)^2}$ for all $i \in V_{(-\tau, \tau)}^{t - 1}$, then output $\pE_{t - 1}$ and terminate. \label{step:term}
\item Otherwise, pick $i \in V_{(-\tau, \tau)}^{t - 1}$ such that $|V_{(-\tau, \tau)}^{t - 1} \setminus C^{t - 1}_{\gamma}(i)| > \frac{n}{\ell(\gamma - \tau^2)^2}$. Compute $\Phi(\pE_{t}|_{X_i=1})$ and $\Phi(\pE_{t}|_{X_i=-1})$ and let $\pE_t$ be equal to the one with larger potential.
\end{enumerate}
\item If the algorithm has not terminated, output NULL. \label{step:null-step}
\end{enumerate}
Notice that, if the algorithm terminates in Step~\ref{step:term}, then the output pseudo-distribution obviously satisfies the condition in Lemma~\ref{lem:main-conditioning}. Hence, we only need to show that the algorithm always terminates in Step~\ref{step:term} (and never reaches Step~\ref{step:null-step}). Recall that we let $\Phi(\pE)$ denote $\sum_{i \in [n]} \pE[X_i]^2$. To prove this, we will analyze the change in $\Phi(\pE_t)$ over time. In particular, we can show the following:
\begin{claim} \label{claim:delta-potential}
For every $t \in [\ell]$, $\Phi(\pE_t) - \Phi(\pE_{t - 1}) > n/\ell$.
\end{claim}

\begin{subproof}[Proof of Claim~\ref{claim:delta-potential}]
First, notice that it suffices to prove the following because $\pE_{t - 1}[(1 - X_i)/2] + \pE_{t - 1}[(1 + X_i)/2] = 1$ and, from our choice of $\pE_t$, we have $\Phi(\pE_t) = \max\{\Phi(\pE_{t - 1}|_{X_i=1}), \Phi(\pE_{t - 1}|_{X_i=-1})\}$.
\begin{align}
\left(\pE\left[\frac{1 - X_i}{2}\right] \cdot \Phi(\pE_{t - 1}|_{X_i=-1}) + \pE\left[\frac{1 + X_i}{2}\right] \cdot \Phi(\pE_{t + 1}|_{X_i=1})\right) - \Phi(\pE_{t - 1}) > n/\ell.
\end{align}

Recall that, from our definition of $\Phi$, the left hand side above can simply be written as
\begin{align} \label{eq:term-by-term}
\sum_{j \in [n]} \left(\pE\left[\frac{1 - X_i}{2}\right] \left(\pE_{t - 1}|_{X_i=-1}[X_j]\right)^2 + \pE\left[\frac{1 + X_i}{2}\right] \left(\pE_{t - 1}|_{X_i=1}[X_j]\right)^2 - \pE_{t - 1}[X_j]^2\right).
\end{align}

From Proposition~\ref{prop:diff-cond}, this is equal to
\begin{align}
\sum_{j \in [n]} \frac{\left(\pE_{t - 1}[X_iX_j] - \pE_{t - 1}[X_i]\pE_{t - 1}[X_j]\right)^2}{1 - \pE_{t - 1}[X_i]^2} &\geqs \sum_{j \in V_{(-\tau, \tau)}^{t - 1} \setminus C^{t - 1}_{\gamma}(i)} \frac{\left(\pE_{t - 1}[X_iX_j] - \pE_{t - 1}[X_i]\pE_{t - 1}[X_j]\right)^2}{1 - \pE_{t - 1}[X_i]^2}\\
&> \sum_{j \in V_{(-\tau, \tau)}^{t - 1} \setminus C^{t - 1}_{\gamma}(i)} \left(\gamma - \tau^2\right)^2 \\
\left(\text{From } |V_{(-\tau, \tau)}^{t - 1} \setminus C^{t - 1}_{\gamma}(i)| \geqs \frac{n}{\ell(\gamma - \tau^2)^2}\right) &> n/\ell,
\end{align}
where the second inequality follows from $|\pE[X_i]|, |\pE[X_j]| < \tau$ and $|\pE[X_iX_j]| > \gamma$ for all $j \in V_{(-\tau, \tau)}^{t - 1} \setminus C^{t - 1}_{\gamma}(i)$.
\end{subproof}

It is now easy to see that Claim~\ref{claim:delta-potential} implies that the algorithm never reaches Step~\ref{step:null-step}. Otherwise, we would have $\Phi(\pE_\ell) > n/\ell + \Phi(\pE_{\ell - 1}) > \cdots > n + \Phi(\pE) > n$, a contradiction.
\end{proof}




\section{The Algorithms} \label{sec:alg}

All of our algorithms follow the same three-step blueprint, as summarized below.

{\bf Step I: Solving for Degree-$n/2^{\Omega(r^2)}$ Pseudo-Expectation.} We first consider the system of constraints corresponding to the best known existing polynomial time algorithm for each problem, and we solve for degree-$n/2^{\Omega(r^2)}$ pseudo-expectation for such a system.

{\bf Step II: Conditioning to Get ``Hollow'' Solution.} Then, we apply Lemma~\ref{lem:main-conditioning} to arrive at a degree-2 pseudo-expectation that satisfies the system and that additionally is hollow, i.e., $|V_{(-\tau, \tau)} \setminus C_{\gamma}(i)| \leqs 2^{r^2}$ for all $i \in V_{(-\tau, \tau)}$ for appropriate values of $\tau, \gamma$. Recall here that $V_{(-\tau, \tau)}$ and $C_{\gamma}(i)$ are defined in Lemma~\ref{lem:main-conditioning}.

{\bf Step III: Following the Existing Algorithm.} Finally, we follow the existing polynomial time approximation algorithms (from~\cite{ARV09,Kar09,ACMM05}) to arrive at an approximate solution for the problem of interest. The improvement in the approximation ratio comes from the fact that our pseudo-expectation is now in the ``easy-to-round'' regime, i.e., the ARV Theorem gives separation of $\Omega(1/r)$ for this regime instead of $\Omega(1/\sqrt{\log n})$ for the general regime.

While the last step closely follows the previous known algorithms, there are sometimes subtlety involves (although there is nothing complicated). In particular, the second ``easy-to-round'' case needs not be handled in previous algorithms but have to be dealt with in our case.  

\subsection{Vertex Cover}

\begin{theorem}
For any $r > 1$ (possibly depending on $n$), there exists an $\exp(n/2^{\Omega(r^2)})\poly(n)$-time $\left(2 - \frac{1}{O(r)}\right)$-approximation algorithm for Vertex Cover on $n$-vertex graphs.
\end{theorem}

\begin{proof}
On input graph $G = (V = [n], E)$, the algorithm works as follows.

\paragraph{Step I: Solving for Degree-$n/2^{\Omega(r^2)}$ Pseudo-Expectation.} For every real number $OBJ \in \mathbb{R}$, let $(\cP^{VC}_{G, OBJ}, \cQ^{VC}_{G, OBJ})$ be the following system of polynomial constraints:
\begin{enumerate}
\item (Boolean) For all $i \in [n]$, $X_i^2 - 1 = 0$.
\item (Edge Cover Condition) For all $(i, j) \in E$, $(1 - X_i)(1 - X_j) = 0$.
\item (Triangle Inequalities) For all $i, j, k \in [n]$, 
\begin{align*}
(X_i - X_k)^2 + (X_k - X_j)^2 - (X_i - X_j)^2 \geqs 0, \\
(X_i + X_k)^2 + (X_k + X_j)^2 - (X_i - X_j)^2 \geqs 0, \\
(X_i - X_k)^2 + (X_k + X_j)^2 - (X_i + X_j)^2 \geqs 0, \\
(X_i + X_k)^2 + (X_k - X_j)^2 - (X_i + X_j)^2 \geqs 0. \\
\end{align*}
\item (Objective Bound) $OBJ - \sum_{i \in [n]} (1 + X_i)/2 \geqs 0$.
\end{enumerate}
Let $D := \lceil 1000n/2^{r^2} \rceil + 2$. The algorithm first uses binary search to find the largest $OBJ$ such that there exists a degree-$D$ pseudo-expectation for $(\cP^{VC}_{G, OBJ}, \cQ^{VC}_{G, OBJ})$. Let this value of $OBJ$ be $OBJ^*$, and let $\pE$ be a degree-$D$ pseudo-expectation satisfying $(\cP^{VC}_{G, OBJ^*}, \cQ^{VC}_{G, OBJ^*})$.

Notice that this step of the algorithm takes $O(n/D)^{O(D)}n^{O(1)} = \exp\left(O(nr^2/2^{r^2})\right)n^{O(1)} = \exp\left(n/2^{\Omega(r^2)}\right)n^{O(1)}$ time. Moreover, observe that the integral solution is a solution to the system with $OBJ = OPT$ where $OPT$ is the size of the optimal vertex cover of $G$. Thus, $OBJ^* \leqs OPT$.

\paragraph{Step II: Conditioning to Get ``Hollow'' Solution.} Use Lemma~\ref{lem:main-conditioning} to find an a degree-2 pseudo-expectation $\pE'$ for $(\cP^{VC}_{G, OBJ^*}, \cQ^{VC}_{G, OBJ^*})$ such that for all $i \in V_{(-0.1, 0.1)}$, $|V_{(-0.1, 0.1)} \setminus C_{0.1}(i)| < 2^{r^2}$.

\paragraph{Step III: Following Karakostas's Algorithm.} The last step of our algorithm proceeds exactly in the same manner as Karakostas's~\cite{Kar09}. First, let $\tau := 1/(10Cr)$ where $C > 1$ is a constant to be specified later; observe that $\tau < 0.1$. We divide the vertices into three groups: (i) $i$'s whose $\pE'[X_i] \geqs \tau$, (ii) $i$'s whose $\pE'[X_i] \leqs -\tau$ and (iii) $i$'s with $|\pE'[X_i]| < \tau$. More formally, let $V_{\geqs \tau} = \{i \in [n] \mid \pE'[X_i] \geqs \tau\}$, $V_{\leqs -\tau} = \{i \in [n] \mid \pE'[X_i] \leqs -\tau\}$ and $V_{(-\tau, \tau)} = \{i \in [n] \mid |\pE'[X_i]| < \tau\}$. The key lemma from~\cite{Kar09} translates in our settings to the following claim. 

\begin{claim}
There exists an absolute constant $\delta > 0$ such that, for any sufficiently large constant $C$, $V_{(-\tau, \tau)}$ contains an independent set of size $\delta|V_{(-\tau, \tau)}|$. Moreover, such an independent set can be found (with probability 2/3) in polynomial time.
\end{claim}

\begin{subproof}
Consider any edge $(i, j) \in E$ such that $i, j \in V_{(-\tau, \tau)}$. From the edge cover constraint, we have
\begin{align*}
\pE'[X_iX_j] = \pE'[X_i] + \pE'[X_j] - 1 < 2\tau - 1.
\end{align*}
As a result, for every $(i, j) \in E \cap (V_{(-\tau, \tau)} \times V_{(-\tau, \tau)})$, we have
\begin{align} \label{eq:edge-distance}
\pE'[(X_i + X_j)^2] = 2(1 + \pE'[X_iX_j]) < 4\tau. 
\end{align}

Now, from $\tau > 0.1$, the hollowness guarantee from Step II allows us to invoke the antipodal version of the ARV structural lemma (Corollary~\ref{cor:arv-antipodal}). This gives us subsets $T, T' \subseteq [n]$ such that $|T| + |T'| \geqs \zeta |V_{(-\tau, \tau)}|$ and, for every $i, i' \in T$ and $j, j' \in T'$, we have $\pE[(X_i + X_{i'})^2], \pE[(X_j + X_{j'})^2] \geqs \theta/r$ where $\theta, \zeta > 0$ are both absolute constants (not depending on $C$).

Observe that, for any $C > 1/\theta$, we have $4\tau < \theta/r$; in other words, for such $C$, \eqref{eq:edge-distance} implies that both $T$ and $T'$ are independent sets. Since $|T| + |T'| \geqs \zeta |V_{(-\tau, \tau)}|$, at least one of them must be an independent set of size at least $(\zeta/2) |V_{(-\tau, \tau)}|$, thereby proving the claim with $\delta = \zeta / 2$.
\end{subproof}

Our algorithm finds an independent set $I \subseteq V_{(-\tau, \tau)}$ of size at least $\delta|V_{(-\tau, \tau)}|$ using the claim above. It then outputs the set $V_{\geqs \tau} \cup (V_{(-\tau, \tau)} \setminus I)$. We now analyze the correctness of our algorithm. To see that the algorithm outputs a valid vertex cover of $G$, first observe that from the edge covering condition, if $(i, j) \in E$, then we have
\begin{align*}
\pE'[X_i] + \pE'[X_j] = 1 - \pE'[X_iX_j] = \frac{1}{2} \pE'[X_i^2] + \frac{1}{2} \pE'[X_j^2] - \pE'[X_iX_j] = \frac{1}{2} \pE'[(X_i - X_j)^2] \geqs 0.
\end{align*}
This implies that $V_{\geqs \tau}$ already cover all edges except those whose both endpoints lie in $V_{(-\tau, \tau)}$. Now, since $I$ is an independent set, $(V_{(-\tau, \tau)} \setminus I)$ must indeed cover all edges within $V_{(-\tau, \tau)}$ and, hence, the output set is a valid vertex cover.

Finally, we will argue that the output solution is of size at most $(2 - 1/O(r)) \cdot OBJ^* \leqs (2 - 1/O(r)) \cdot OPT$. To see this, first observe that
\begin{align*}
|V_{\geqs \tau}| \leqs \frac{2}{1 + \tau} \sum_{i \in V_{\geqs \tau}} \pE'[(1 + X_i)/2] = \left(2 - \frac{1}{O(r)}\right) \sum_{i \in V_{\geqs \tau}} \pE'[(1 + X_i)/2] 
\end{align*}
Next, suppose that we choose $C$ such that $1/(10C) < \delta/2$, we have
\begin{align*}
|V_{(-\tau, \tau)} \setminus I| \leqs (1 - \delta) \cdot |V_{(-\tau, \tau)}| &\leqs (1 - \delta) \cdot \frac{2}{1 - \tau} \sum_{i \in V_{(-\tau, \tau)}} \pE'[(1 + X_i)/2] \\
(\text{From our choice of } C) &\leqs (1 - \delta) \cdot \frac{2}{1 - \delta/2} \sum_{i \in V_{(-\tau, \tau)}} \pE'[(1 + X_i)/2] \\
&\leqs (2 - \delta) \sum_{i \in V_{(-\tau, \tau)}} \pE'[(1 + X_i)/2].
\end{align*}
By summing the two inequalities, we have
\begin{align*}
|V_{\geqs \tau} \cup (V_{(-\tau, \tau)} \setminus I)| &\leqs \left(2 - \frac{1}{O(r)}\right) \sum_{i \in V_{\geqs \tau}} \pE'[(1 + X_i)/2]  + (2 - \delta) \sum_{i \in V_{(-\tau, \tau)}} \pE'[(1 + X_i)/2] \\
&\leqs \left(2 - \frac{1}{O(r)}\right) \cdot OBJ^*,
\end{align*}
which concludes our proof.
\end{proof}

\subsection{Balanced Separator}

\begin{theorem}
For any $r > 1$ (possibly depending on $n$), there exists an $\exp(n/2^{\Omega(r^2)})\poly(n)$-time $O(r)$-approximation for Balanced Separator on $n$-vertex graphs.
\end{theorem}

\begin{proof}
On input graph $G = (V = [n], E)$, the algorithm works as follows.

\paragraph{Step I: Solving for Degree-$n/2^{\Omega(r^2)}$ Pseudo-Expectation.} For every real number $OBJ \in \mathbb{R}$, let $(\cP^{BS}_{G, OBJ}, \cQ^{BS}_{G, OBJ})$ be the following system of equations:
\begin{enumerate}
\item (Boolean) For all $i \in [n]$, $X_i^2 - 1 = 0$.
\item (Balance) $\sum_{i, j \in [n]} (X_i - X_j)^2 - 16n^2/9 \geqs 0,n/3 - \sum_{i \in [n]} X_i \geqs 0$ and $\sum_{i \in [n]} X_i + n/3 \geqs 0$.
\item (Triangle Inequalities) For all $i, j, k \in [n]$, 
\begin{align*}
(X_i - X_k)^2 + (X_k - X_j)^2 - (X_i - X_j)^2 \geqs 0 \\
(1 - X_i)^2 + (1 - X_j)^2 - (X_i - X_j)^2 \geqs 0, \\
(1 + X_i)^2 + (1 - X_j)^2 - (X_i - X_j)^2 \geqs 0.
\end{align*}
\item (Objective Bound) $4 \cdot OBJ - \sum_{(i, j) \in E} (X_i - X_j)^2 \geqs 0$.
\end{enumerate}
Let $D := \lceil 1000n/2^{r^2} \rceil + 2$. The algorithm first uses binary search to find the largest $OBJ$ such that there exists a degree-$D$ pseudo-expectation for $(\cP^{BS}_{G, OBJ}, \cQ^{BS}_{G, OBJ})$. Let this value of $OBJ$ be $OBJ^*$, and let $\pE$ be a degree-$D$ pseudo-expectation satisfying $(\cP^{BS}_{G, OBJ^*}, \cQ^{BS}_{G, OBJ^*})$.

Again, observe that this step takes $\exp\left(n/2^{\Omega(r^2)}\right)\poly(n)$ time and $OBJ^* \leqs OPT$ where $OPT$ is the number of edges cut in the balanced separator of $G$.

\paragraph{Step II: Conditioning to Get ``Hollow'' Solution.} Use Lemma~\ref{lem:main-conditioning} to find an a degree-2 pseudo-expectation $\pE'$ for $(\cP^{BS}_{G, OBJ^*}, \cQ^{BS}_{G, OBJ^*})$ such that for all $i \in V_{(-0.9, 0.9)}$, $|V_{(-0.9, 0.9)} \setminus C_{0.9}(i)| < 2^{r^2}$.

\paragraph{Step III: Following ARV Algorithm.}
The last step follows the ARV algorithm~\cite{ARV09}. The first step in the algorithm is to use the structural lemma to obtain two large well separated set. While in the traditional setting, the structural theorem can be applied immediately; we have to be more careful and treat the two ``easy-to-round'' cases differently. This is formalized below.

\begin{claim}
There exist disjoint subsets $T, T' \subseteq [n]$ that are $\Omega(1/r)$-separated. Moreover, these subsets can be found (with probability 2/3) in polynomial time.
\end{claim}

\begin{subproof}
Similar to before, for every $a, b \in \mathbb{R}$, let $V_{\geqs a} = \{i \in [n] \mid \pE'[X_i] \geqs a\}$, $V_{\leqs b} = \{i \in [n] \mid \pE'[X_i] \leqs b\}$ and $V_{(a, b)} = \{i \in [n] \mid |\pE'[X_i]| < \tau\}$.

Let $\tau = 0.9$. We consider the following two cases:
\begin{enumerate}
\item $|V_{\geqs \tau}| \geqs 0.1n$ or $|V_{\leqs -\tau}| \geqs 0.1n$. Suppose without loss of generality that it is the former. We claim that $|V_{\leqs 0.8}| \geqs 0.2n$. To see that this is the case, observe that
\begin{align*}
n/3 \geqs \sum_{i \in [n]} \pE'[X_i] 
&= \sum_{i \in V \setminus V_{\leqs 0.8}} \pE'[X_i] + \sum_{i \in V_{\leqs 0.8}} \pE'[X_i] \\
&\geqs 0.8(n - |V_{\leqs 0.8}|) - |V_{\leqs 0.8}| \\
&= 0.8n - 1.8|V_{\leqs 0.8}| 
\end{align*}
which implies that $|V_{\leqs 0.8}| \geqs (0.8n - n/3)/1.8 > 0.2n$ as desired. Let $T = V_{\geqs \tau}$ and $T' = V_{\leqs 0.8}$. As we have shown, $|T|, |T'| \geqs \Omega(n)$. Moreover, for every $i \in T$ and $j \in T'$, triangle inequality implies that
\begin{align*}
\pE'[X_iX_j] \leqs \pE'[X_iX_j] + \pE'[(1 - X_i)(1 + X_j)] = 1 - \pE'[X_i] + \pE'[X_j] < 1 - 0.9 + 0.8 = 0.9.
\end{align*}
That is, we have $\pE'[(X_i - X_j)^2] = 2 - 2\pE'[X_iX_j] > 0.2$, completing the proof for the first case.
\item $|V_{\geqs \tau}| < 0.1n$ and $|V_{\leqs -\tau}| < 0.1n$. This implies that $|V_{(-\tau, \tau)}| \geqs 0.8n$. Moreover, observe that 
\begin{align*}
\sum_{i, j \in V_{(-\tau, \tau)}} \pE[(X_i - X_j)^2] &= \sum_{i, j \in [n]} \pE[(X_i - X_j)^2] - \sum_{\substack{i, j \in [n] \\ i \in V_{\geqs \tau} \text{ or } j \in V_{\geqs \tau}}} \pE[(X_i - X_j)^2] - \sum_{\substack{i, j \in [n] \\ i \in V_{\leqs \tau} \text{ or } j \in V_{\leqs \tau}}} \pE[(X_i - X_j)^2] \\
&\geqs 16n^2/9 - 8n|V_{\geqs \tau}| - 8n|V_{\leqs \tau}| \\
&> 0.1n^2.
\end{align*}

Hence, applying the ARV Structural Theorem (Theorem~\ref{thm:arv-sos}) to $V_{(-\tau, \tau)}$ yields the desired $T, T'$.
\end{enumerate}
Thus, in both cases, we can find the desired $T, T'$ in randomized polynomial time.
\end{subproof}

Once we have found the sets $T, T'$, we use the following rounding scheme from~\cite{LR99,ARV09}:
\begin{itemize}
\item Pick $\theta$ uniformly at random from $[0, d(T, T'))$.
\item Let $S = \{i \in [n] \mid d(i, T) < \theta\}$.
\item Output $(S, V \setminus S)$.
\end{itemize}
Observe that $T \subseteq S$ and $T' \subseteq (V \setminus S)$, which means that $|S|, |V \setminus S| > \Omega(n)$; in other words, the output is a valid (pseudo-)solution for the Balanced Separator Problem. Moreover, for every $(i, j) \in E$, it is easy to see that the probability that the two endpoints end up in different sets is at most $|d(i, T) - d(j, T)|/d(T, T') \leqs d(i, j)/d(T, T') \leqs O(r) \cdot d(i, j)$. As a result, the expected number of edges cut by our solution is $O(r) \cdot \sum_{(i, j) \in E} d(i, j) = O(r) \cdot OBJ^*$, which completes our proof.
\end{proof}

\subsection{Uniform Sparsest Cut}

\begin{theorem}
For any $r > 1$ (possibly depending on $n$), there exists an $\exp(n/2^{\Omega(r^2)})\poly(n)$-time $O(r)$-approximation for Uniform Sparsest Cut on $n$-vertex graphs.
\end{theorem}

\begin{proof}
Given an input graph $G = (V = [n], E)$. For every $t \in [n]$, let us use $\Phi_G(t)$ to denote $\min_{\substack{S \subseteq V \\ |S| = t}} \Phi_G(S)$. For each $t \in [n]$, we will design an algorithm so that it outputs a set $S \subseteq V$ with $\Phi_G(S) \leqs O(r) \cdot \Phi_G(t)$. By running this algorithm for every $t \in [n]$ and output the set with minimum edge expansion, we can approximate the Uniform Sparsest Cut to within $O(r)$ factor.

Let us now fix $t \in [n]$. Observe that we may assume w.l.o.g. that $t \leqs n/2$. Moreover, when $t \geqs n/2^{r^2/100}$, we can just enumerate all subsets $S \subseteq V$ of size $t$ and find the one with smallest edge expansion; this is an exact algorithm for $\Phi_G(t)$ that runs in time $(n/t)^{O(t)} n^{O(1)} = \exp(n/2^{\Omega(r^2)}) n^{O(1)}$. Hence, from this point onwards, we may assume that $\kappa := t/n$ lies in $(1/2^{r^2/100}, 1/2]$. We will also assume without loss of generality that $r \leqs \sqrt{\log n}/100$; otherwise, ARV algorithm~\cite{ARV09} already gives the desired approximation in polynomial time.

\paragraph{Step I: Solving for Degree-$n/2^{\Omega(r^2)}$ Pseudo-Expectation.} For every real number $OBJ \in \mathbb{R}$, let $(\cP^{SC}_{G, OBJ, t}, \cQ^{SC}_{G, OBJ, t})$ be the following system of equations:
\begin{enumerate}
\item (Boolean) For all $i \in [n]$, $X_i^2 - 1 = 0$.
\item (Size) $\sum_{i, j \in [n]} (X_i - X_j)^2 - 8t(n - t) = 0$.
\item (Triangle Inequalities) For all $i, j, k \in [n]$, 
\begin{align*}
(X_i - X_k)^2 + (X_k - X_j)^2 - (X_i - X_j)^2 \geqs 0 \\
(1 - X_i)^2 + (1 - X_j)^2 - (X_i - X_j)^2 \geqs 0, \\
(1 + X_i)^2 + (1 - X_j)^2 - (X_i - X_j)^2 \geqs 0.
\end{align*}
\item (Objective Bound) $OBJ \cdot t - \sum_{(i, j) \in E} (X_i - X_j)^2 \geqs 0$.
\end{enumerate}
Let $D := \lceil 1000n/2^{r^2} \rceil + 2$. The algorithm first uses binary search to find the largest $OBJ$ such that there exists a degree-$D$ pseudo-expectation for $(\cP^{SC}_{G, OBJ, t}, \cQ^{SC}_{G, OBJ, t})$. Let this value of $OBJ$ be $OBJ^*$, and let $\pE$ be a degree-$D$ pseudo-expectation satisfying $(\cP^{SC}_{G, OBJ^*, t}, \cQ^{SC}_{G, OBJ^*, t})$.

Again, observe that this step takes $\exp\left(n/2^{\Omega(r^2)}\right)\poly(n)$ time and $OBJ^* \leqs \Phi_G(t)$.

\paragraph{Step II: Conditioning to Get ``Hollow'' Solution.} Let $\tau = 1 - \kappa/10$ and let $\gamma = 1 - \kappa/30$. Applying Lemma~\ref{lem:main-conditioning} to $\pE$ gives us a degree-2 pseudo-expectation $\pE'$ for $(\cP^{SC}_{G, OBJ, t^*}, \cQ^{SC}_{G, OBJ, t^*})$ such that, for all $i \in V_{(-\tau, \tau)}$, we have 
$$|V_{(-\tau, \tau)} \setminus C_{\gamma}(i)| \leqs \frac{n}{(1000n/2^{r^2})(\gamma - \tau^2)^2} = O\left(\frac{2^{r^2}}{\kappa^2}\right) \leqs O\left(\frac{2^{r^2}}{2^{-r^2/50}}\right) = O(2^{2r^2})$$
where the second inequality follows from $\kappa \geqs 2^{-r^2/100}$.

\paragraph{Step III: Following ARV Algorithm.}
The last step follows the ARV algorithm~\cite{ARV09} for Sparsest Cut. Here we will adhere to the notation of Lee~\cite{Lee05}; in fact, the proof below is exactly the same as that of Lee with only one exception: we have to consider the second ``easy-to-round'' case, which will be the first case in the lemma below. For convenient, we will write $d(i, \emptyset)$ to denote $\pE[(X_i - 1)^2]$; due to the fact that we add triangle inequalities for $1$ as well, $d$ still remains a valid metric on $[n] \cup \{\emptyset\}$. The key lemma of~\cite{ARV09,Lee05} is the following.

\begin{lemma} \label{lem:set-pairwise}
There exists a set $T \subseteq [n]$ such that
\begin{align} \label{eq:set-spread}
\sum_{i, j \in [n]} |d(i, T) - d(j, T)| \geqs \frac{1}{O(r)} \sum_{i, j \in [n]} d(i, j) = \Omega(nt/r).
\end{align}
\end{lemma}

\begin{subproof}
We consider the following two cases:
\begin{enumerate}
\item $|V_{\geqs \tau}| \geqs 0.1n$ or $|V_{\leqs \tau}| \geqs 0.1n$. Assume without loss of generality that it is the former. Let $T = V_{\geqs \tau}$. We have
\begin{align*}
4\kappa n^2 \leqs 8t(n - t) = \sum_{i, j \in [n]} d(i, j) \leqs \sum_{i, j \in [n]} \left(d(i, \emptyset) + d(j, \emptyset)\right) = 2n\sum_{i \in [n]} d(i, \emptyset) \leqs 2n\sum_{i \in [n]} \left(d(i, T) + \kappa\right)
\end{align*}
where the last inequality comes from the fact that $d(i, T) \subseteq \cB(\emptyset, \kappa)$.

The above inequality implies that $\sum_{i \in [n]} d(i, T) \geqs 2\kappa n$. As a result, we have
\begin{align*}
\sum_{i, j \in [n]} |d(i, T) - d(j, T)| \geqs \sum_{j \notin T, i \in T} d(i, T) \geqs 2\kappa n|T| \geqs 0.2 \kappa n^2 > \Omega(1) \cdot \sum_{i, j \in [n]} d(i, j)
\end{align*}
as desired.
\item $|V_{\geqs \tau}| < 0.1n$ and $|V_{\leqs \tau}| < 0.1n$. In this case, we have $|V_{(-\tau, \tau)}| > 0.8n$. Observe that
\begin{align*}
8\kappa n^2 \geqs \sum_{i, j \in [n]} d(i, j)
&\geqs \sum_{i, j \in V_{(-\tau, \tau)}} d(i, j) \\
&\geqs \sum_{i \in V_{(-\tau, \tau)}} \sum_{j \in V_{(-\tau, \tau)} \setminus \cB(i, 100\kappa)} d(i, j) \\
&\geqs \sum_{i \in V_{(-\tau, \tau)}} 100\kappa \cdot |V_{(-\tau, \tau)} \setminus \cB(i, 100\kappa)|.
\end{align*}
Hence, there must exist $i^* \in V_{(-\tau, \tau)}$ such that $|V_{(-\tau, \tau)} \setminus \cB(i^*, 100\kappa)| \leqs 8\kappa n^2 / (100 \kappa \cdot 0.8n) = 0.1n$. In other words, $|V_{(-\tau, \tau)} \cap \cB(i^*, 100\kappa)| > 0.7n$. Let us consider the set $U = V_{(-\tau, \tau)} \cap \cB(i^*, 100\kappa)$. Observe that the guarantee of Step II implies that, for all $i \in U$, we have $|U \setminus \cB(i, \kappa/15)| \leqs O(2^{2r^2})$. Recall that $r \leqs \sqrt{\log n}/100$, meaning that $|U \setminus \cB(i, \kappa/15)| < 0.5n$ for sufficiently large $n$. As a result, we have
\begin{align*}
\sum_{i, j \in U} d(i, j) &\geqs \sum_{i \in U} \sum_{j \in U \setminus \cB(i, \kappa/15)} d(i, j) \\
&\geqs \sum_{i \in U} (\kappa/15) \cdot |U \setminus \cB(i, \kappa/15)| \\
&\geqs 0.7n \cdot (\kappa/15) \cdot 0.2n \\
&\geqs \Omega(\kappa n^2).
\end{align*}

Notice that the metric space $(U, d)$ has diameter $O(\kappa)$ and that it is $\Omega(\kappa |U|^2)$-separated. Hence, we can now apply the ARV Theorem (Theorem~\ref{thm:arv-metric}\footnote{Notice that here all distances are scaled by a factor of $\kappa$ and hence the $\Omega(\kappa/r)$-separation.}) on $U$ which gives us the sets $T, T'$ of size $\Omega(|U|) = \Omega(n)$ which are $\Omega(\kappa/r)$-separated. This means that
\begin{align*}
\sum_{i, j \in [n]} |d(i, T) - d(j, T)| \geqs \sum_{j \in T, i \in T'} d(i, T) \geqs \sum_{j \in T, i \in T'} \Omega(\kappa/r) = \Omega(\kappa n^2 / r) = \Omega(1/r) \cdot \sum_{i, j \in [n]} d(i, j).
\end{align*} 
\end{enumerate}
\end{subproof}

Finally, given the set $T$ from Lemma~\ref{lem:set-pairwise}, we consider the following algorithm:
\begin{itemize}
\item Sort vertices by the distance to $T$ in increasing order. Let $\pi(1), \dots, \pi(n)$ be the sorted list.
\item Consider sets $S_\ell := \{\pi(1), \dots, \pi(\ell)\}$ for all $\ell \in [n - 1]$.
\item Output $S_\ell$ that minimizes $\Phi_G(S_\ell)$ among all $\ell \in [n - 1]$.
\end{itemize}
The output set has expansion
\begin{align*}
\min_{\ell \in [n - 1]} \Phi_G(S_\ell) &\leqs 2n \cdot \min_{\ell \in [n - 1]} \frac{E(S_\ell, V \setminus S_\ell)}{2|S_\ell| \cdot |V \setminus S_\ell|} \\
&= 2n \cdot \min_{\ell \in [n - 1]} \frac{\sum_{(i, j) \in E} |\ind[i \in S_\ell] - \ind[j \in S_\ell]|}{\sum_{i, j \in V} |\ind[i \in S_\ell] - \ind[j \in S_\ell]|} \\
&\leqs 2n \left(\frac{\sum_{\ell \in [n - 1]} (d(\pi(\ell + 1), T) - d(\pi(\ell), T)) \cdot \sum_{(i, j) \in E} |\ind[i \in S_\ell] - \ind[j \in S_\ell]|}{\sum_{\ell \in [n - 1]} (d(\pi(\ell + 1), T) - d(\pi(\ell), T)) \cdot \sum_{i, j \in V} |\ind[i \in S_\ell] - \ind[j \in S_\ell]|}\right) \\
&= 2n \left(\frac{\sum_{(i, j) \in E} |d(i, T) - d(j, T)|}{\sum_{(i, j) \in [n]} |d(i, T) - d(j, T)|}\right) \\
&\overset{\eqref{eq:set-spread}}{\leqs} O(r/t) \left(\sum_{(i, j) \in E} |d(i, T) - d(j, T)|\right) \\
&\leqs O(r/t) \cdot \sum_{(i, j) \in E} d(i, j) \\
&\leqs O(r) \cdot OBJ^*,
\end{align*}
which concludes our proof.
\end{proof}

\subsection{Minimum 2CNF Deletion and Minimum UnCut}

\begin{theorem} \label{thm:cnfdel-uncut}
For any $r > 1$ (possibly depending on $n$), there exists an $\exp(n/2^{O(r^2)})\poly(n)$-time $O(r)$-approximation for Minimum 2CNF Deletion and Minimum UnCut where $n$ denote the number of input variables for Minimum 2CNF Deletion and the number of input vertices for Minimum UnCut.
\end{theorem}

The algorithms of~\cite{ACMM05} for both Min 2CNF Deletion and Min UnCut are derived via an algorithm for a more general problem called \emph{Minimum Symmetric Directed Cut} as defined below.

{\bf Min Symmetric DiCut.} A directed graph $G = (V, E)$ is said to be \emph{symmetric} if (i) the vertex set $V$ is $[-n] \cup [n]$ where $[-n] = \{-n, \dots, -1\}$ and (ii) an arc $(i, j)$ belongs to $E$ if and only if $(-j, -i)$ also belongs to $E$. For every $S \subseteq V$, we use $-S$ to denote $\{-i \mid i \in S\}$. A set $S$ is said to be symmetric iff $S = -S$. Furthermore, we say that a cut $(S, V \setminus S)$ is symmetric iff $V \setminus S = -S$.

In Min Symmetric DiCut, we are given as an input a symmetric directed graph $G = (V, E)$ and the goal is to find a symmetric cut $(S, -S)$ that minimizes the number of arcs going from $S$ to $-S$.

\begin{proposition}[{\cite{ACMM05}}] \label{prop:sym-dicut}
If there exists a $T(n)$-time $\rho(n)$-approximation for Min Symmetric DiCut, then there also exists $O(T(n))$-time $\rho(n)$-approximation for Min 2CNF Deletion and Min UnCut.
\end{proposition}

\begin{proof}[Proof Sketch of Proposition~\ref{prop:sym-dicut}]
We can reduce Min 2CNF Deletion to Min Symmetric DiCut as follows. Suppose that the variable set in the Min 2CNF Deletion input is $\cX = \{x_1, \dots, x_n\}$. Let the input graph of Min Symmetric DiCut be $G = (V, E)$ such that $V = [-n] \cup [n]$ and, for each clause which is an OR of two literals $b_1$ and $b_2$, we add two arcs $(-\sgn(b_1) \cdot \var(b_1), \sgn(b_2) \cdot \var(b_2))$ and $(\sgn(b_2) \cdot \var(b_2), -\sgn(b_1) \cdot \var(b_1))$ where $\sgn(b) \in \{\pm 1\}$ is -1 iff the clause $b$ is a negation of a variable and $\var(b) \in [n]$ denote the index of the variable corresponding to $b$. Observe that there is a one-to-one correspondence between assignments from $\cX$ to $\{0, 1\}$ and symmetric cuts in $G$ such that the number of arcs cut is exactly twice the number of clauses unsatisfied. Hence, a $T(n)$-time $\rho(n)$-approximation algorithm for Min Symmetric DiCut translates directly to an $O(T(n))$-time $\rho(n)$-approximation algorithm for Min 2CNF Deletion.

The reduction from Min UnCut to Min Symmetric DiCut is similar. Suppose that the input graph to Min UnCut is $G' = (V', E')$ where $V' = [n]$. Then, we create the input graph $G = (V, E)$ for Min Symmetric Cut where $V = [-n] \cup [n]$ and, for each edge $\{i, j\} \in E'$, we add two arcs $(-i, j)$ and $(-j, i)$ to $E$. Analogous to before, it is simple to see that there is a one-to-one correspondence between cuts of $G'$ and symmetric cuts of $G$ such that the number of arcs cut in $G$ is exactly twice the number of uncut edges in $G'$. Thus, a $T(n)$-time $\rho(n)$-approximation algorithm for Min Symmetric DiCut implies to an $O(T(n))$-time $\rho(n)$-approximation algorithm for Min UnCut.
\end{proof}

Given Proposition~\ref{prop:sym-dicut}, we can focus our attention to design an approximation algorithm for Min Symmetric DiCut. In particular, to show Theorem~\ref{thm:cnfdel-uncut}, it suffices to prove the following:

\begin{theorem} \label{thm:sym-dicut}
For any $r > 1$ (possibly depending on $n$), there exists an $\exp(n/2^{O(r^2)})\poly(n)$-time $O(r)$-approximation for Min Symmetric DiCut.
\end{theorem}

To prove the theorem, it will be convenient to define the notion of \emph{symmetric directed metric} used in the work of Agarwal \etal~\cite{ACMM05}. The notations surrounding symmetric directed metric are defined in an analogous fashion to those of metric (Definition~\ref{def:metric}):

\begin{definition}[Directed Metric-Related Notions]
A \emph{symmetric directed metric} $d$ on a symmetric set $X \subseteq [-n] \cup [n]$ is a distance function $d: X \times X \to \mathbb{R}_{\geqs 0}$ that satisfies (1) $d(x, x) = 0$, (2) symmetry $d(x, y) = d(-y, -x)$ and (3) triangle inequality $d(x, z) \leqs d(x, y) + d(y, z)$, for all $x, y, z \in X$. We use the following notations throughout this section:
\begin{itemize}
\item For $x \in X$ and $S, T \subseteq X$, $d(x, S) := \min_{y \in S} d(x, y)$ and $d(S, T) := \min_{y \in S} d(y, T)$. 
\item We say that $S, T$ are \emph{$\Delta$-separated} iff $d(S, T) \geqs \Delta$.
\item An \emph{(open) ball of radius $r$ around $x$} denoted by $\cB_d(x, r)$ is defined as $\{y \in X \mid d(x, y) < r\}$. 
\item A metric space $(X, d)$ is said to be \emph{($r$, $m$)-hollow} if $|\cB_d(x, r)| \leqs m$ for all $x \in X$.
\end{itemize}
\end{definition}

We will need an additional notation of \emph{volume} of a set of vertices which is simply the total distance of all edges with both endpoints lie in the set:

\begin{definition}[Volume]
Given a directed graph $G = (V, E)$ where $V \subseteq [n] \cup [-n]$ is a symmetric set and a symmetric directed metric $d$ on $V$, the volume of $M \subseteq V$ is defined as $\vol_{d, G}(M) := \sum_{(i, j) \in E \atop i, j \in M} d(i, j)$. 
\end{definition}

Similar to the case of (undirected) metric above, boolean degree-2 pseudo-expectation naturally induces a symmetric directed metric on $[-n] \cup [n]$, as specified below.

\begin{definition}\label{prop:pE-directed-metric}
Let $\pE: \mathbb{R}_2[X_1, \dots, X_n] \to \mathbb{R}$ be any degree-2 pseudo-expectation that satisfies the following triangle inequalities for all $i, j, k \in [n]$: $\pE[(X_i - X_j)^2] \leqs \pE[(X_i - X_k)^2] + \pE[(X_k - X_j)^2], \pE[(X_i - X_j)^2] \leqs \pE[(X_i + X_k)^2] + \pE[(X_k + X_j)^2], \pE[(X_i + X_j)^2] \leqs \pE[(X_i - X_k)^2] + \pE[(X_k + X_j)^2]$ and 
$\pE[(X_i + X_j)^2] \leqs \pE[(X_i + X_k)^2] + \pE[(X_k - X_j)^2]$. 

Define $d^{\text{dir}}_{\pE}: ([n] \cup [-n]) \times ([n] \cup [-n]) \to \mathbb{R}_{\geqs 0}$ by $d^{\text{dir}}_{\pE}(i, j) = \pE[(1 + \sgn(i) \cdot X_{|i|})(1 - \sgn(j) \cdot X_{|j|})]$. Then, $d^{\text{dir}}_{\pE}$ is a symmetric directed metric on $[-n] \cup [n]$.
\end{definition}

Finally, we state the version of the ARV Lemma used in the symmetric directed metric case. This version is closely related to the antipodal version of the ARV Lemma (Corollary~\ref{cor:arv-antipodal}), with two exceptions: (1) the ``size'' of $S$ is not measured in terms of $|S|$ but rather in $\vol(S)$ and (2) the distance is now in terms of the directed metric instead of the usual metric distance. For a full proof of how to derive such a variant from the standard version, please refer to Lemma 4.6 of~\cite{ACMM05}.

\begin{lemma}[ARV Lemma: Directed Metric Version~\cite{ACMM05}] \label{lem:arv-directed}
Let $\pE: \mathbb{R}_2[X_1, \dots, X_n] \to \mathbb{R}$ be any degree-2 pseudo-expectation that satisfies the following conditions for any $\beta > 0$ and $m \in \N$:
\begin{itemize}
\item (Boolean) For every $i \in [n]$, $\pE[X_i^2] = 1$. 
\item (Triangle Inequality) For every $i, j, k \in [n]$, 
\begin{align*}
\pE[(X_i - X_j)^2] \leqs \pE[(X_i - X_k)^2] + \pE[(X_k - X_j)^2], \\
\pE[(X_i - X_j)^2] \leqs \pE[(X_i + X_k)^2] + \pE[(X_k + X_j)^2], \\
\pE[(X_i + X_j)^2] \leqs \pE[(X_i - X_k)^2] + \pE[(X_k + X_j)^2], \\
\pE[(X_i + X_j)^2] \leqs \pE[(X_i + X_k)^2] + \pE[(X_k - X_j)^2].
\end{align*}
\item (Hollowness) For all $i \in [n]$, $|\{j \in [n] \mid |\pE[X_iX_j]| > 1 - \beta\}| \leqs m$.
\end{itemize}
Let $G = ([n] \cup [-n], E)$ be any graph and $M \subseteq [-n] \cup [n]$ be any symmetric set. Then, there exists a randomized polynomial time algorithm that, with probability 2/3, produces a subset $S \subseteq M$ such that $\frac{\vol_{d^{\text{dir}}_{\pE}, G}(M \setminus (S \cup -S))}{\vol_{d^{\text{dir}}_{\pE}, G}(M)} \leqs 1 - \Omega_{\beta}(1)$ and $d^{\text{dir}}_{\pE}(S, -S) \geqs \Omega_{\beta}(1/\sqrt{\log m})$.
\end{lemma}

With all the preliminaries in place, we proceed to prove Theorem~\ref{thm:sym-dicut}. As with the previous proofs, to ease the notations, we will drop the subscripts when the graph, metric or pseudo-distribution are already clear from the context.

\begin{proof}[Proof of Theorem~\ref{thm:sym-dicut}]
On input graph $G = (V = [-n] \cup [n], E)$, the algorithm works as follows.

\paragraph{Step I: Solving for Degree-$n/2^{\Omega(r^2)}$ Pseudo-Expectation.} For every real number $OBJ \in \mathbb{R}$, let $(\cP^{VC}_{G, OBJ}, \cQ^{VC}_{G, OBJ})$ be the following system of polynomial constraints:
\begin{enumerate}
\item (Boolean) For all $i \in [n]$, $X_i^2 - 1 = 0$.
\item (Triangle Inequalities) For all $i, j, k \in [n]$, 
\begin{align*}
(X_i - X_k)^2 + (X_k - X_j)^2 - (X_i - X_j)^2 \geqs 0, \\
(X_i + X_k)^2 + (X_k + X_j)^2 - (X_i - X_j)^2 \geqs 0, \\
(X_i - X_k)^2 + (X_k + X_j)^2 - (X_i + X_j)^2 \geqs 0, \\
(X_i + X_k)^2 + (X_k - X_j)^2 - (X_i + X_j)^2 \geqs 0. \\
\end{align*}
\item (Objective Bound) $OBJ - \sum_{(i, j) \in E \cap ([n] \times [n])} (1 - X_iX_j) + \sum_{(i, j) \in E \cap ([-n] \times [n])} (1 + X_iX_j) \geqs 0$.
\end{enumerate}
Let $D := \lceil 1000n/2^{r^2} \rceil + 2$. The algorithm first uses binary search to find the largest $OBJ$ such that there exists a degree-$D$ pseudo-expectation for $(\cP^{VC}_{G, OBJ}, \cQ^{VC}_{G, OBJ})$. Let this value of $OBJ$ be $OBJ^*$, and let $\pE$ be a degree-$D$ pseudo-expectation satisfying $(\cP^{VC}_{G, OBJ^*}, \cQ^{VC}_{G, OBJ^*})$.

Notice that this step of the algorithm takes $O(n/D)^{O(D)}n^{O(1)} = \exp\left(O(nr^2/2^{r^2})\right)n^{O(1)} = \exp\left(n/2^{\Omega(r^2)}\right)n^{O(1)}$ time. Observe also that the integral solution is a solution with $OBJ = OPT$ where $OPT$ is the minimum number of arcs cut by any symmetric cut of $G$. Thus, $OBJ^* \leqs OPT$.

\paragraph{Step II: Conditioning to Get ``Hollow'' Solution.} Use Lemma~\ref{lem:main-conditioning} to find an a degree-2 pseudo-expectation $\pE'$ for $(\cP^{VC}_{G, OBJ^*}, \cQ^{VC}_{G, OBJ^*})$ such that for all $i \in V_{(-0.1, 0.1)}$, $|V_{(-0.1, 0.1)} \setminus C_{0.1}(i)| < 2^{r^2}$.

\paragraph{Step III: Following Agarwal \etal's Algorithm.} The last step of our algorithm proceeds exactly in the same manner as Agarwal \etal's~\cite{ACMM05}. The algorithm proceed in iterations as follows.
\begin{enumerate}
\item First, initialize $M_0 \leftarrow [n] \cup [-n]$ and $\ell \leftarrow 0$.
\item While $M_\ell$ is not empty, execute the following:
\begin{enumerate}
\item If $\ell = 0$, let $S_0 = \{i \mid i \in [n] \text{ and } \pE'[X_i] \leqs -0.1\} \cup \{-i \mid i \in [n] \text{ and } \pE'[X_i] \geqs 0.1\}$. If $S_0 = \emptyset$, then let $\ell \leftarrow 1$ and skip the following steps.
\item Otherwise, use Lemma~\ref{lem:arv-directed} to find a set $S_\ell \subseteq M_\ell$ such that $\vol(M \setminus (S_\ell \cup -S_{\ell})) \leqs (1 - C) \cdot \vol(M_\ell)$ for some constant $C > 0$ and $d(S_\ell, -S_\ell) \geqs \Omega(1/r)$.
\item Pick $\theta$ uniformly at random from $[0, d(S_\ell, -S_\ell)/2)$.
\item Let $T_\ell = \{i \in [-n] \cup [n] \mid d(S_\ell, i) \leqs \theta\}$, $M_{\ell + 1} \leftarrow M_\ell \setminus (T_\ell \cup -T_\ell)$ and $\ell \leftarrow \ell + 1$.
\end{enumerate}
\item Output the cut $(S, -S)$ where $S = \cup_{\ell \geqs 0} T_\ell$.
\end{enumerate}

To bound the expected number of arcs cut, first observe that $\vol(M_\ell)$ shrinks by a factor of $(1 - C)$ in each iteration, i.e., $\vol(M_\ell) \leqs (1 - C)^{\ell - 1} \cdot \vol(V)$. Next, consider the arcs cut in the $\ell$-th step for $\ell \geqs 1$, i.e., the arcs $(i, j)$ that lies in $(T_\ell \times M_\ell) \cup (M_{\ell} \times -T_{\ell})$. Consider any arc $(i, j) \in M_\ell \times M_\ell$. The probability that the arc is cut in the $\ell$-th iteration is at most $d(i, j)/(d(S_\ell, -S_{\ell})/2) \leqs O(r) \cdot d(i, j)$. Hence, in total the expected number of arcs cut in this iteration is at most
\begin{align*}
\sum_{(i, j) \in E \cap (M_{\ell} \times M_{\ell})} O(r) \cdot d(i, j) = O(r) \cdot \vol(M_\ell) \leqs O(r) \cdot (1 - C)^{\ell - 1} \cdot \vol(V).
\end{align*}
As a result, the expected total number of arcs cut in all iterations $\ell \geqs 1$ is at most $\sum_{\ell \geqs 1} O(r) \cdot (1 - C)^{\ell - 1} \cdot \vol(V) \leqs O(r) \cdot \vol(V)$.

It can be similarly argued that the expected number of arcs cut in the first step is $O(\vol(V))$. Thus, the expected total number of arcs cut is $O(r) \cdot \vol(V)$. Finally, observe that the objective bound can be written as $OBJ \geqs \vol(V)/2$. As a result, this yields an $O(r)$-approximation for the problem.
\end{proof}

\section{Conclusion and Open Questions} \label{sec:open}

In this work, we use the conditioning framework in the SoS Hierarchy together with the ARV Structural Theorem to design ``fast'' exponential time approximation algorithms for Vertex Cover, Uniform Sparsest Cut and related problems that achieve significant speed-up over the trivial ``limited brute force'' algorithms. While we view this as a step towards ultimately understanding the time vs approximation ratio trade-off for these problems, many questions remain open.

First and most importantly, as discussed in the introduction, current lower bounds do not rule out subexponential time approximation algorithms in the regime of our study. For instance, an 1.9-approximation algorithm for Vertex Cover could still possibly be achieved in say $2^{O(\sqrt{n})}$ time. Similarly for Uniform Sparsest Cut and Balanced Separator, $O(1)$-approximation for them could still possibly be achieved in subexponential time. The main open question is to either confirm that such algorithms exist, or rule them out under certain believable complexity hypotheses.

Another, perhaps more plausible, direction is to try to extend our technique to other problems for which the best known polynomial time approximation algorithms employ the ARV Structural Theorem. This includes Balanced Vertex Separator, (Non-uniform) Sparsest Cut, and Minimum Linear Arrangement. While the first problem admits $O(\sqrt{\log n})$-approximation in polynomial time~\cite{FHL08}, several more ingredients beyond the ARV Theorem are required to make the algorithm work. On the other hand, the latter two problems only admit $O(\sqrt{\log n} \log \log n)$-approximation~\cite{ALN05,CHKR10,FL07}. It seems challenging to remove this $\log \log n$ factor and achieve a constant factor approximation, even in our ``fast'' exponential time regime.

\bibliography{main}

\newcommand{\etalchar}[1]{$^{#1}$}
\begin{thebibliography}{KMOW17}

\bibitem[ABG13]{ABG13}
Per Austrin, Siavosh Benabbas, and Konstantinos Georgiou.
\newblock Better balance by being biased: {A} 0.8776-approximation for max
  bisection.
\newblock In {\em SODA}, pages 277--294, 2013.

\bibitem[ABS15]{ABS15}
Sanjeev Arora, Boaz Barak, and David Steurer.
\newblock Subexponential algorithms for unique games and related problems.
\newblock {\em J. {ACM}}, 62(5):42:1--42:25, 2015.

\bibitem[ACMM05]{ACMM05}
Amit Agarwal, Moses Charikar, Konstantin Makarychev, and Yury Makarychev.
\newblock {$O(\sqrt{\log n})$} approximation algorithms for min {UnCut}, min
  {2CNF} deletion, and directed cut problems.
\newblock In {\em STOC}, pages 573--581, 2005.

\bibitem[AIMS10]{AIMS10}
Sanjeev Arora, Russell Impagliazzo, William Matthews, and David Steurer.
\newblock Improved algorithms for unique games via divide and conquer.
\newblock {\em ECCC}, 17:41, 2010.

\bibitem[ALM{\etalchar{+}}98]{ALMSS98}
Sanjeev Arora, Carsten Lund, Rajeev Motwani, Madhu Sudan, and Mario Szegedy.
\newblock Proof verification and the hardness of approximation problems.
\newblock {\em J. {ACM}}, 45(3):501--555, 1998.

\bibitem[ALN05]{ALN05}
Sanjeev Arora, James~R. Lee, and Assaf Naor.
\newblock Euclidean distortion and the sparsest cut.
\newblock In {\em STOC}, pages 553--562, 2005.

\bibitem[ARV09]{ARV09}
Sanjeev Arora, Satish Rao, and Umesh~V. Vazirani.
\newblock Expander flows, geometric embeddings and graph partitioning.
\newblock {\em J. {ACM}}, 56(2):5:1--5:37, 2009.

\bibitem[AS98]{AS98}
Sanjeev Arora and Shmuel Safra.
\newblock Probabilistic checking of proofs: {A} new characterization of {NP}.
\newblock {\em J. {ACM}}, 45(1):70--122, 1998.

\bibitem[BBH{\etalchar{+}}12]{BBHKSZ12}
Boaz Barak, Fernando G. S.~L. Brand{\~{a}}o, Aram~Wettroth Harrow, Jonathan~A.
  Kelner, David Steurer, and Yuan Zhou.
\newblock Hypercontractivity, sum-of-squares proofs, and their applications.
\newblock In {\em STOC}, pages 307--326, 2012.

\bibitem[BCL{\etalchar{+}}17]{BCLNN17}
Nikhil Bansal, Parinya Chalermsook, Bundit Laekhanukit, Danupon Nanongkai, and
  Jesper Nederlof.
\newblock New tools and connections for exponential-time approximation.
\newblock {\em CoRR}, abs/1708.03515, 2017.

\bibitem[BEP11]{BEP11}
Nicolas Bourgeois, Bruno Escoffier, and Vangelis~Th. Paschos.
\newblock Approximation of max independent set, min vertex cover and related
  problems by moderately exponential algorithms.
\newblock {\em Discrete Applied Mathematics}, 159(17):1954--1970, 2011.

\bibitem[BK09]{BK09}
Nikhil Bansal and Subhash Khot.
\newblock Optimal long code test with one free bit.
\newblock In {\em FOCS}, pages 453--462, 2009.

\bibitem[BRS11]{BRS11}
Boaz Barak, Prasad Raghavendra, and David Steurer.
\newblock Rounding semidefinite programming hierarchies via global correlation.
\newblock In {\em FOCS}, pages 472--481, 2011.

\bibitem[BS14]{BS14}
Boaz Barak and David Steurer.
\newblock Sum-of-squares proofs and the quest toward optimal algorithms.
\newblock {\em ECCC}, 21:59, 2014.

\bibitem[BYE85]{BE85}
R.~Bar-Yehuda and S.~Even.
\newblock A local-ratio theorem for approximating the weighted vertex cover
  problem.
\newblock In G.~Ausiello and M.~Lucertini, editors, {\em Analysis and Design of
  Algorithms for Combinatorial Problems}, volume 109 of {\em North-Holland
  Mathematics Studies}, pages 27 -- 45. North-Holland, 1985.

\bibitem[CHKR10]{CHKR10}
Moses Charikar, Mohammad~Taghi Hajiaghayi, Howard~J. Karloff, and Satish Rao.
\newblock $\ell_2^2$ spreading metrics for vertex ordering problems.
\newblock {\em Algorithmica}, 56(4):577--604, 2010.

\bibitem[CKK{\etalchar{+}}06]{CKKRS06}
Shuchi Chawla, Robert Krauthgamer, Ravi Kumar, Yuval Rabani, and D.~Sivakumar.
\newblock On the hardness of approximating multicut and sparsest-cut.
\newblock {\em Computational Complexity}, 15(2):94--114, 2006.

\bibitem[CMM10]{CMM10}
Moses Charikar, Konstantin Makarychev, and Yury Makarychev.
\newblock Local global tradeoffs in metric embeddings.
\newblock {\em {SIAM} J. Comput.}, 39(6):2487--2512, 2010.

\bibitem[DKK{\etalchar{+}}16]{DKKMS16}
Irit Dinur, Subhash Khot, Guy Kindler, Dor Minzer, and Muli Safra.
\newblock Towards a proof of the 2-to-1 games conjecture?
\newblock {\em ECCC}, 23:198, 2016.

\bibitem[DKK{\etalchar{+}}17]{DKKMS17}
Irit Dinur, Subhash Khot, Guy Kindler, Dor Minzer, and Muli Safra.
\newblock On non-optimally expanding sets in grassmann graphs.
\newblock {\em ECCC}, 24:94, 2017.

\bibitem[DS05]{DS05}
Irit Dinur and Shmuel Safra.
\newblock On the hardness of approximating minimum vertex cover.
\newblock {\em Annals of Mathematics}, 162(1):439--485, 2005.

\bibitem[Fei02]{Fei02}
Uriel Feige.
\newblock Relations between average case complexity and approximation
  complexity.
\newblock In {\em STOC}, pages 534--543, 2002.

\bibitem[FHL08]{FHL08}
Uriel Feige, MohammadTaghi Hajiaghayi, and James~R. Lee.
\newblock Improved approximation algorithms for minimum weight vertex
  separators.
\newblock {\em {SIAM} J. Comput.}, 38(2):629--657, 2008.

\bibitem[FL07]{FL07}
Uriel Feige and James~R. Lee.
\newblock An improved approximation ratio for the minimum linear arrangement
  problem.
\newblock {\em Inf. Process. Lett.}, 101(1):26--29, 2007.

\bibitem[GJ79]{GJ79}
Michael~R. Garey and David~S. Johnson.
\newblock {\em Computers and Intractability: {A} Guide to the Theory of
  NP-Completeness}.
\newblock W. H. Freeman, 1979.

\bibitem[Gri01]{Gri01}
Dima Grigoriev.
\newblock Complexity of positivstellensatz proofs for the knapsack.
\newblock {\em Computational Complexity}, 10(2):139--154, 2001.

\bibitem[GS11]{GS11}
Venkatesan Guruswami and Ali~Kemal Sinop.
\newblock Lasserre hierarchy, higher eigenvalues, and approximation schemes for
  graph partitioning and quadratic integer programming with {PSD} objectives.
\newblock In {\em FOCS}, pages 482--491, 2011.

\bibitem[GVY96]{GVY96}
Naveen Garg, Vijay~V. Vazirani, and Mihalis Yannakakis.
\newblock Approximate max-flow min-(multi)cut theorems and their applications.
\newblock {\em {SIAM} J. Comput.}, 25(2):235--251, 1996.

\bibitem[Hal02]{Hal02}
Eran Halperin.
\newblock Improved approximation algorithms for the vertex cover problem in
  graphs and hypergraphs.
\newblock {\em {SIAM} J. Comput.}, 31(5):1608--1623, 2002.

\bibitem[H{\aa}s01]{Has01}
Johan H{\aa}stad.
\newblock Some optimal inapproximability results.
\newblock {\em J. {ACM}}, 48(4):798--859, 2001.

\bibitem[IP01]{IP01}
Russell Impagliazzo and Ramamohan Paturi.
\newblock On the complexity of k-{SAT}.
\newblock {\em J. Comput. Syst. Sci.}, 62(2):367--375, 2001.

\bibitem[IPZ01]{IPZ01}
Russell Impagliazzo, Ramamohan Paturi, and Francis Zane.
\newblock Which problems have strongly exponential complexity?
\newblock {\em J. Comput. Syst. Sci.}, 63(4):512--530, 2001.

\bibitem[Kar09]{Kar09}
George Karakostas.
\newblock A better approximation ratio for the vertex cover problem.
\newblock {\em {ACM} Trans. Algorithms}, 5(4):41:1--41:8, 2009.

\bibitem[Kho02]{Kho02}
Subhash Khot.
\newblock On the power of unique 2-prover 1-round games.
\newblock In {\em CCC}, page~25, 2002.

\bibitem[Kho06]{Kho06}
Subhash Khot.
\newblock Ruling out {PTAS} for graph min-bisection, dense k-subgraph, and
  bipartite clique.
\newblock {\em {SIAM} J. Comput.}, 36(4):1025--1071, 2006.

\bibitem[KKMO07]{KKMO07}
Subhash Khot, Guy Kindler, Elchanan Mossel, and Ryan O'Donnell.
\newblock Optimal inapproximability results for {MAX-CUT} and other 2-variable
  csps?
\newblock {\em {SIAM} J. Comput.}, 37(1):319--357, 2007.

\bibitem[KMOW17]{KMOW17}
Pravesh~K. Kothari, Ryuhei Mori, Ryan O'Donnell, and David Witmer.
\newblock Sum of squares lower bounds for refuting any {CSP}.
\newblock In {\em Proceedings of the 49th Annual {ACM} {SIGACT} Symposium on
  Theory of Computing, {STOC} 2017, Montreal, QC, Canada, June 19-23, 2017},
  pages 132--145, 2017.

\bibitem[KMS17]{KMS17}
Subhash Khot, Dor Minzer, and Muli Safra.
\newblock On independent sets, 2-to-2 games, and grassmann graphs.
\newblock In {\em STOC}, pages 576--589, 2017.

\bibitem[KMS18]{KMS18}
Subhash Khot, Dor Minzer, and Muli Safra.
\newblock Pseudorandom sets in grassmann graph have near-perfect expansion.
\newblock {\em ECCC}, 25:6, 2018.

\bibitem[KPRT97]{KPRT97}
Philip~N. Klein, Serge~A. Plotkin, Satish Rao, and {\'{E}}va Tardos.
\newblock Approximation algorithms for steiner and directed multicuts.
\newblock {\em J. Algorithms}, 22(2):241--269, 1997.

\bibitem[KR08]{KR08}
Subhash Khot and Oded Regev.
\newblock Vertex cover might be hard to approximate to within 2-epsilon.
\newblock {\em J. Comput. Syst. Sci.}, 74(3):335--349, 2008.

\bibitem[KV15]{KV15}
Subhash Khot and Nisheeth~K. Vishnoi.
\newblock The unique games conjecture, integrality gap for cut problems and
  embeddability of negative-type metrics into $\ell_1$.
\newblock {\em J. {ACM}}, 62(1):8:1--8:39, 2015.

\bibitem[Las02]{Las02}
Jean~B. Lasserre.
\newblock An explicit equivalent positive semidefinite program for nonlinear
  0-1 programs.
\newblock {\em {SIAM} Journal on Optimization}, 12(3):756--769, 2002.

\bibitem[Lee05]{Lee05}
James~R. Lee.
\newblock On distance scales, embeddings, and efficient relaxations of the cut
  cone.
\newblock In {\em SODA}, pages 92--101, 2005.

\bibitem[LR99]{LR99}
Frank~Thomson Leighton and Satish Rao.
\newblock Multicommodity max-flow min-cut theorems and their use in designing
  approximation algorithms.
\newblock {\em J. {ACM}}, 46(6):787--832, 1999.

\bibitem[MR10]{MR10}
Dana Moshkovitz and Ran Raz.
\newblock Two-query {PCP} with subconstant error.
\newblock {\em J. {ACM}}, 57(5):29:1--29:29, 2010.

\bibitem[MR16]{MR16}
Pasin Manurangsi and Prasad Raghavendra.
\newblock A birthday repetition theorem and complexity of approximating dense
  {CSP}s.
\newblock {\em CoRR}, abs/1607.02986, 2016.

\bibitem[MS85]{MS85}
Burkhard Monien and Ewald Speckenmeyer.
\newblock Ramsey numbers and an approximation algorithm for the vertex cover
  problem.
\newblock {\em Acta Inf.}, 22(1):115--123, 1985.

\bibitem[Nes00]{Nes00}
Yurii Nesterov.
\newblock Squared functional systems and optimization problems.
\newblock In {\em High performance optimization}, pages 405--440. Springer,
  2000.

\bibitem[O'D17]{O17}
Ryan O'Donnell.
\newblock {SOS} is not obviously automatizable, even approximately.
\newblock In {\em ITCS}, pages 59:1--59:10, 2017.

\bibitem[OZ13]{OZ13}
Ryan O'Donnell and Yuan Zhou.
\newblock Approximability and proof complexity.
\newblock In {\em SODA}, pages 1537--1556, 2013.

\bibitem[Par00]{Par00}
Pablo~A. Parrilo.
\newblock {\em Structured semidefinite programs and semialgebraic geometry
  methods in robustness and optimization}.
\newblock PhD thesis, California Institute of Technology, 2000.

\bibitem[PY91]{PY91}
Christos~H. Papadimitriou and Mihalis Yannakakis.
\newblock Optimization, approximation, and complexity classes.
\newblock {\em J. Comput. Syst. Sci.}, 43(3):425--440, 1991.

\bibitem[RS10]{RS10}
Prasad Raghavendra and David Steurer.
\newblock Graph expansion and the unique games conjecture.
\newblock In {\em STOC}, pages 755--764, 2010.

\bibitem[RST12]{RST12}
Prasad Raghavendra, David Steurer, and Madhur Tulsiani.
\newblock Reductions between expansion problems.
\newblock In {\em CCC}, pages 64--73, 2012.

\bibitem[RT12]{RT12}
Prasad Raghavendra and Ning Tan.
\newblock Approximating {CSP}s with global cardinality constraints using {SDP}
  hierarchies.
\newblock In {\em SODA}, pages 373--387, 2012.

\bibitem[Sch08]{Sch08}
Grant Schoenebeck.
\newblock Linear level lasserre lower bounds for certain k-csps.
\newblock In {\em FOCS}, pages 593--602, 2008.

\bibitem[YZ14]{YZ14}
Yuichi Yoshida and Yuan Zhou.
\newblock Approximation schemes via sherali-adams hierarchy for dense
  constraint satisfaction problems and assignment problems.
\newblock In {\em ITCS}, pages 423--438, 2014.

\end{thebibliography}
\bibliographystyle{alpha}

\end{document}